\documentclass[letterpaper,11pt]{article}
\pdfoutput=1
  \usepackage{fullpage}

\usepackage{mdframed}
\usepackage{thm-restate}
\definecolor{light-gray}{gray}{0.95}
\usepackage{mathtools}

 \usepackage{hyperref}
\usepackage{url}
 \usepackage{xcolor,color,url,eurosym,graphicx,amssymb,amsfonts,amsmath,algorithm,algorithmic,bbm,xspace}
 \usepackage[most]{tcolorbox}
 \usepackage{empheq}
  
\newtcbox{\mymath}[1][]{%
    nobeforeafter, math upper, tcbox raise base,
    enhanced, colframe=blue!30!black,
    colback=blue!30, boxrule=1pt,
    #1}

\newtheorem{problem}{Problem}
\newcommand\myeq{\mathrel{\stackrel{\makebox[0pt]{\mbox{\normalfont\tiny def}}}{=}}}
\newtheorem{proposition}{Proposition}
\newcommand{\diag}{\mathop\mathrm{diag}}

  \DeclareMathOperator{\Tr}{Tr}

 \makeatletter
\newcommand*{\rom}[1]{\expandafter\@slowromancap\romannumeral #1@}
\makeatother

\newcommand{\hide}[1]{} 
\newtheorem{claim}{Claim}

\newenvironment{proof}{\par\noindent\textit{Proof.}}{$\Box$\par\bigskip\par}

\newtheorem{theorem}{Theorem}

\newtheorem{observation}{Observation}
\newtheorem{corollary}{Corollary}

\newcommand{\beql}[1]{\begin{equation}\label{#1}}

\newcommand{\beq}[1]{\begin{equation}\label{#1}}
\newcommand{\eeq}{\end{equation}}

\newcommand{\field}[1]{\mathbb{#1}} 

\newcommand{\spara}[1]{\smallskip\noindent{\bf #1}}

\newcommand{\squishlist}{
 \begin{list}{$\bullet$}
  {  \setlength{\itemsep}{0pt}
     \setlength{\parsep}{3pt}
     \setlength{\topsep}{3pt}
     \setlength{\partopsep}{0pt}
     \setlength{\leftmargin}{2em}
     \setlength{\labelwidth}{1.5em}
     \setlength{\labelsep}{0.5em}
} }
\newcommand{\squishlisttight}{
 \begin{list}{$\bullet$}
  { \setlength{\itemsep}{0pt}
    \setlength{\parsep}{0pt}
    \setlength{\topsep}{0pt}
    \setlength{\partopsep}{0pt}
    \setlength{\leftmargin}{2em}
    \setlength{\labelwidth}{1.5em}
    \setlength{\labelsep}{0.5em}
} }

\newcommand{\squishdesc}{
 \begin{list}{}
  {  \setlength{\itemsep}{0pt}
     \setlength{\parsep}{3pt}
     \setlength{\topsep}{3pt}
     \setlength{\partopsep}{0pt}
     \setlength{\leftmargin}{1em}
     \setlength{\labelwidth}{1.5em}
     \setlength{\labelsep}{0.5em}
} }

\newcommand{\squishend}{
  \end{list}
}

\newcommand{\squishlistt}{
 \begin{list}{---}
  {  \setlength{\itemsep}{0pt}
     \setlength{\parsep}{3pt}
     \setlength{\topsep}{3pt}
     \setlength{\partopsep}{0pt}
     \setlength{\leftmargin}{2em}
     \setlength{\labelwidth}{1.5em}
     \setlength{\labelsep}{0.5em}
} }

\begin{document}

\title{Minimizing Polarization and Disagreement in Social Networks}

\author{
Cameron Musco \\ MIT \\ cnmusco@mit.edu
\and Christopher Musco \\ MIT \\ cpmusco@mit.edu
\and  Charalampos E. Tsourakakis \\ Boston University \\ ctsourak@bu.edu
}

\date{}
\begin{titlepage}
\maketitle
\begin{abstract}
The rise of social media and online social networks has been a disruptive force in society. Opinions are increasingly shaped by interactions on online social media, and social phenomena including disagreement and polarization are now tightly woven into everyday life.  In this work we initiate the study of the following question: 

 \begin{quotation}
\noindent Given $n$ agents, each with its own initial opinion that reflects its core value on a topic, and an opinion dynamics model, what is the structure of a social network that minimizes {\em polarization} and {\em disagreement} simultaneously?
\end{quotation}
  
\noindent This question is  central to recommender systems: should a recommender system prefer a link suggestion between two online users with similar mindsets in order to keep disagreement low, or between two users with different opinions in order to expose each to the other's viewpoint of the world, and decrease overall levels of polarization?  Such decisions have an important global effect on society \cite{williams2007social}.   Our contributions  include a mathematical  formalization of  this question as an optimization problem and an exact, time-efficient algorithm. We also prove that there always exists a network with $O(n/\epsilon^2)$ edges that is a $(1+\epsilon)$ approximation to the optimum. Our formulation is an instance of optimization over {\em graph topologies}, also considered e.g., in \cite{boyd2004fastest,daitch2009fitting,sun2006fastest}. For a fixed graph, we additionally show how to optimize our objective function over the agents' innate opinions in polynomial time.  

We perform an empirical study  of our proposed methods on synthetic and real-world data that verify their value as mining tools to better understand the trade-off between  of disagreement and polarization.  We find that there is a lot of space to reduce both polarization and disagreement in real-world networks; for instance, on a Reddit network where users exchange comments on politics, our methods achieve a $\sim 60\,000$-fold reduction in  polarization and disagreement. Our code is available at  \url{https://github.com/tsourolampis/polarization-disagreement}.

\end{abstract}

\end{titlepage}

\section{Introduction}
\label{sec:intro}
The enormous popularity of social media and online social networks has led to fundamental changes in how humans share and form opinions.
Several events have stirred  fierce debates online, such as the internal memo of James Damore at Google that leaked online and showed how society and users of social media become strongly polarized around controversial issues \cite{damore}. Social phenomena such as disagreement and polarization that have existed in human societies for millenia, are now taking place in an online virtual world,   with a huge  impact  on society.   Furthermore,  opinions are increasingly shaped by social media.  For this reason  Facebook's recommender systems were recently accused of influencing voters by propagating fake news and Russian ads during the U.S. national elections \cite{nytimes}, thus stirring additional controversy about the role of tech giants in society.  Google has faced similar accusations about influencing the outcome of elections \cite{sciencemag}.   

From a business perspective, recommender systems aim to maximize revenue -- some in the short-term, others in  the long-term -- by optimizing mathematical objectives  
related to  click-through rate and user engagement. However, given the increasing power of tech-giants like Facebook and Google, a better understanding of  the effect of recommended links on society is required \cite{nytimes2}.

The following facts motivate  our work. Valdis Krebs  analyzed purchasing trends on Amazon. He found that  during the 2008 presidential elections,  people who already supported Barack Obama tended to be the same people buying books that painted him in a positive light. Similarly, people who disliked him, bought books that painted him in a negative light \cite[Chapter 1]{kadushin2012understanding}. This bias is known as {\em confirmation bias} \cite{kahneman2011thinking}, and lies at the root of the spread of  conspiracy theories and fake news. Put simply, most humans  avoid challenging their opinions. Therefore, a recommender system trained on real-data, whose goal is to maximize revenue and increase user engagement, may naturally end up creating ``echo-chambers''. Equivalently, the recommended links minimize the {\em disagreement} that the user experiences. In the context of social networks, connections between users with similar mindsets are preferred over connections between users with different mindsets. 

On the other hand, minimizing disagreement leads to greater {\em polarization}. To see intuitively why this is the case, consider a topic with two prevalent opinions such as supporting either Democrat or Republican politics. As users connect to users with similar mindsets, two clusters with strong and extreme opinions are formed, leading to greater polarization between the two groups \cite{boxell2017internet}. This polarization has harmful effects on society; for example, reaching political consensus becomes much harder.   This trade-off between {\em disagreement} and {\em polarization} motivates our work. We are interested in understanding the structure of  a social network that minimizes both disagreement and polarization.  We introduce and study two key problems, summarized below. For both problems, we use as our underlying opinion dynamics model the  Friedkin-Johnsen   model \cite{friedkin1990social} that includes both disagreement and consensus, as it associates   each node with an innate 
opinion and an expressed opinion. For details see Section~\ref{sec:related}.

\spara{$\bullet$ Minimizing Polarization and Disagreement Over Graphs.} We initiate the study of the following important question in understanding the global effect of  recommended connections in social networks on society.   We formalize  this problem mathematically, and provide a polynomial time algorithm by proving that our optimization formulation is convex. Our formulation is an instance  of optimization over {\em graph topologies}. For other such optimization problems, see \cite{boyd2004fastest,daitch2009fitting,sun2006fastest}.

\begin{tcolorbox}
\begin{problem}
\label{prob1} 
Given $n$ agents, each with its own initial opinion that reflects its core value on a topic, and an opinion dynamics model, what is the structure of a connected social network with a given total edge weight that minimizes {\em polarization } and {\em disagreement} simultaneously?
\end{problem}
\end{tcolorbox}

\noindent  We additionally prove that there always exists a  graph with at most $O(n/\epsilon^2)$ edges that achieves a $(1+\epsilon)$-approximation to the above problem. That is, disagreement and polarization can be minimized by very sparsely connected networks.
  
\spara{$\bullet$  Minimizing Polarization and Disagreement Over Opinions.} We also study the following optimization problem. Here, the social network is given and we wish to modify the agents'  innate opinions  to minimize  polarization and disagreement at equilibrium. This problem aims to understand the effect of targeted ads or recommendations designed to influence innate opinions.

\begin{tcolorbox}
\begin{problem}
\label{prob3} 
\noindent Given a  network $G$ on $n$ agents, each with its own initial opinion, an opinion dynamics model, and a budget $\alpha>0$,  how should we change the initial opinions using total opinion mass at most $\alpha$ in order to minimize polarization and disagreement?
\end{problem}
\end{tcolorbox}

\spara{$\bullet$ Experiments.} We evaluate our methods on synthetic and real data. Our findings indicate that our graph optimization methods  result in significantly lower levels of polarization and disagreement on Twitter and Reddit, and confirm that our proposed methods can be used as mining tools to better understand the effects of link recommendation. We also observe experimentally that existing graph topologies are far away from optimizing polarization and disagreement. For instance, our proposed method shows that we can obtain a $6.2 \times 10^4$-fold reduction in polarization and disagreement by optimizing the graph topology of our Reddit dataset.

\spara{Roadmap.}   The paper is organized as follows: Section~\ref{sec:related} presents related work.
Section~\ref{sec:proposed}  presents our algorithmic contributions.
Section~\ref{sec:exp} presents experimental findings on synthetic and real-world data.  Finally, Section~\ref{sec:concl} concludes the paper.

\spara{Notation.}  We use the following notation throughout our paper. Let $G(V,E,w)$ be a weighted connected undirected graph, with $V=[n]$ and  $|E|=m$.  Let $N(i)$ be the neighborhood of node $i$ and  $d(i)=\sum\limits_{j \in N(i)} w_{ij}$ be its degree.  Let $A$ be the adjacency matrix, and $L \myeq D-A$  be the combinatorial Laplacian. Here  $D=diag(d(1),\ldots,d(n))$ is a diagonal matrix with the degrees in its diagonal. Each agent $i \in V$ has an innate opinion $s_i \in [0,1]$. Let $s=(s_1,\ldots,s_n) \in [0,1]^n$ be the vector of innate opinions.   Finally, let $\vec{1},\vec{0}$ be the all-ones and all-zeros vectors respectively.

\section{Related Work}
\label{sec:related}
To the best of our knowledge, we are the first to formally define and study the problems suggested in Section~\ref{sec:intro}.  In the following we review related work. 

\spara{Modeling opinion dynamics.} Opinion dynamics has been a topic of intense  study by political scientists, economists, sociologists, physicists, control systems scientists, and computer scientists, among others.  Opinion dynamics model social learning processes. Such models have important applications in understanding  political  voting (e.g., \cite{braha2017voting,acemoglu2011opinion}), viral marketing (e.g., \cite{Das:2014:MOD:2556195.2559896}), and various other phenomena that take place in social media (e.g., \cite{quattrociocchi2014opinion}).  There exist {\em discrete} and {\em continuous} models. In the former, agents have only two possible opinions, 0 or 1. In the latter, opinions span a continuous range; popular choices in the literature are intervals $[0,1]$, and $[-1,1]$. The voter model is a popular discrete model that  was originally described by  Clifford and Sudbury \cite{clifford1973model}  in the context of a spatial conflict between animals for territory. In general, discrete models apply  various interaction mechanisms that may update an agent's opinion.  Such mechanisms include randomly adopting the opinion of a connected neighbor or applying a local majority rule \cite{durrett1988lecture}.    

DeGroot introduced a continuous opinion dynamics model in his seminal work on consensus formation \cite{degroot1974reaching}. A set of $n$ individuals in society start with initial opinions on a subject. Individual opinions are updated using the average of the neighborhood of a fixed social network. Friedkin and Johnsen \cite{friedkin1997social} extended the DeGroot model to include both disagreement and consensus by mixing each individual's \emph{innate belief} with some weight into the averaging process.  We discuss  this model in greater detail below.    Other  continuous models are the Hegselmann-Krause model \cite{hegselmann2002opinion} and the Deffuant-Weisbuch model  \cite{zhang2013opinion}. 
Das, Gollapudi, Munagala use a model which combines the continuous and discrete approach to fit opinion data better \cite{Das:2014:MOD:2556195.2559896}.  Bhawalkar, Gollapudi, Munagala  study the scenario where the network and the agents' opinions co-evolve  \cite{bhawalkar2013coevolutionary}. 

It is worth mentioning that many aspects  of opinion formation have been incorporated into existing models such as opinion leaders, external influences and many others.  We refer the interested reader to the  excellent survey  by Mossel and Tamuz and references therein for more related work on opinion dynamics models \cite{mossel2017opinion}.

\spara{Friedkin-Johnsen opinion dynamics.}    Each node maintains a persistent internal (or innate) opinion $s_i$, that remains constant. The node updates its expressed opinion $z_i$ through repeated averaging.  More precisely, if $w_{ij} \geq 0$ is the weight on edge $(i,j) \in E$ and $N(i)$ denotes the neighborhood of node $i$, then in one time step agent $i$ updates its  opinion to be the average

$$ z_i = \frac{s_i + \sum\limits_{j \in N(i)} w_{ij} z_j}{1+\sum\limits_{j \in N(i)} w_{ij}}.$$ 

\noindent Let $z^*$ be the equilibrium of this process. The value $z^*_i$ at equilibrium is the {\em expressed opinion} of node $i$. It is well-known that the equilibrium  $z^*$  is the solution to a linear system of equations: 

\begin{tcolorbox}
\vspace{-.5em}
\begin{align}
\label{eq:fj}
\begin{array}{c@{}cr}
z^* =  &(I+L)^{-1} s &  \text{[Friedkin-Johnsen Equilibrium]}
\end{array}
\end{align}
\end{tcolorbox}

\noindent Here $L$ is the combinatorial Laplacian of the connection graph $G$.  Notice that $(I+L)$ is always invertible since it is positive definite (in fact, all of its eigenvalues are greater than or equal to 1).   The Friedkin-Johnsen model has been used by Bindel et al. as a model for understanding  the price of anarchy in society when individuals selfishly update their opinions in order to minimize the stress they experience \cite{bindel2015bad}. Here, stress consists of two terms, the stress due to the fact that node $i$ at equilibrium may express a different opinion than its innate one, and the difference of its expressed opinion and the expressed opinions of its neighbors. Formally, the  
stress of a node $i$ is defined as $(z^*_i - s_i)^2 + \sum\limits_{j \in N(i)} w_{ij}(z^*_i-z^*_j)^2$.  Gionis, Terzi, and Tsaparas use this model to suggest a variation of influence maximization \cite{gionis2013opinion}.

\spara{Polarization.}   Munson et al.   have created a browser widget that measures the bias of users based on the news articles they read, and suggests articles of the opposing view to reduce polarization  \cite{munson2013encouraging}.  Liao and Fu provide a tool that makes users aware of the extremity of their opinion, and how well justified it is given their expertise on the topic \cite{liao2014can}. Dandekar et al. define a polarizing  opinion formation process that increases a disagreement index at the end of the process, show limitations of existing linear models to capture extreme polarization, and suggest a non-linear model \cite{dandekar2013biased}. 
Flaxman et al. study political ``echo-chambers'' by examining browsing histories and show several interesting empirical findings with respect to how social networks affect humans' opinions \cite{flaxman2016filter}. 

Closest to our work lies that of Matakos et al. \cite{matakos2017measuring}.   In their work, they focus on the $\ell_2$ norm of the equilibrium vector $z^*$ under the Friedkin-Johnsen model, which is close to our measure of polarization, introduced in Section \ref{subsec:contro}, except that they do not mean center the opinion vector.
Given a parameter $k$, they consider the problem of choosing a set of $k$ nodes such that if their innate opinions are set to zero, the proposed polarization index will be minimized. This problem is NP-hard, and focuses just on polarization, without considering disagreement.
Close to our work at a conceptual level also lies the work of Garimella, Morales, Gionis, and Mathioudakis \cite{garimella2016quantifying}, which focuses on Twitter data. Garimella et al. create graphs from Twitter data, propose a new definition of polarization, and detect topics  that cause intense debates based on this measure.  In follow-up work, they consider the problem of making link recommendations to reduce polarization \cite{garimella2017reducing,garimella2017exposing}. 

Finally, an important aspect of using opinion dynamics with real data is opinion mining. For example, how do we map tweets to opinions, i.e., real numbers? A lot of  work in this area has been done by the natural language processing (NLP) community; see the recent survey 
by Nakov for more details  on sentiment analysis and Twitter opinion mining \cite{nakov2017semantic}. 

\spara{Optimizing over graph topologies.}  Problem 1 (see Section~\ref{sec:intro})  aims to optimize an objective over graph topologies. We review some instances of related graph topology optimization problems, most of which lie outside the graph mining literature, which has focused on edge recommendation rather than global structural  results.  

 In 2009, Daitch, Kelner, and Spielman suggested learning graphs that fit vector points well in order to leverage graph mining techniques for machine learning \cite{daitch2009fitting}.  They propose using the objective   $||LX||_F^2$ where $L$ is the graph Laplacian   and $X$ is a data matrix representing $n$ points in $d$ dimensions \cite{daitch2009fitting}. Their  objective  minimizes the reconstruction error over all points, where the reconstruction error of a point is the $\ell_2$ norm of the difference of the point and the weighted average of its neighbors. Other groups have considered related approaches; see for instance \cite{dong2016learning,kalofolias2016learn}.  

The central problem in phylogenetics is to learn an evolutionary tree that explains the observed existence of $n$ species. A popular family of techniques are distance-based;  find an evolutionary tree (whose leaves correspond to the species)  which most closely matches a set of ${n \choose 2}$ pairwise genetic distances  \cite{felsenstein1985confidence}. In machine learning, a wide variety of techniques  (e.g., Chow-Liu) have been introduced to learn graphical models from data via optimization problems over graph topologies \cite{koller2009probabilistic}.  Boyd et al. have studied the problem of finding the fastest mixing Markov chain on  a graph \cite{boyd2004fastest,sun2006fastest}, casting it as an optimization problem over all possible weighted subgraphs  of a given graph.

\section{Proposed Method}
\label{sec:proposed}
\subsection{Polarization and Disagreement} 
\label{subsec:contro}
Following \eqref{eq:fj}, let $z^* = (I+L)^{-1} s$ be the equilibrium vector of opinions according to the Friedkin-Johnsen model, for a social network $G(V,E,w)$, and innate opinions $s:V\rightarrow [0,1]$. 

\spara{Disagreement.}   We define the disagreement  $d(u,v)$ of  edge $(u,v)$ as the squared difference between the opinions of $u,v$ at equilibrium:     $d(u,v) \myeq w_{uv} (z^*_u-z^*_v)^2$. We define total disagreement $D_{G,s}$ as:

\begin{tcolorbox}
\vspace{-.5em}
\begin{align}
\label{eq:dfn_disagree}
\begin{array}{c@{}cr}
D_{G,s} \myeq &  \sum\limits_{(u,v) \in E} d(u,v). &  \text{  [Disagreement]}
\end{array}
\end{align}
\end{tcolorbox}

\spara{Polarization.} Intuitively, polarization should measure how opinions at equilibrium deviate from the average. There are many ways to quantify this. We choose the  standard definition of variance, i.e., the second moment   of the opinions.  Specifically, let $ \bar{z} $ be the mean-centered equilibrium vector:  

$$ \bar{z} =z^*  - \frac{z^{*T}  \vec{1}}{n} \vec{1}.$$
\noindent Then the polarization $P_{G,s}$ is defined to be:
\begin{tcolorbox}
\vspace{-.5em}
\begin{align}
\label{eq:dfn_contr}
\begin{array}{c@{}cr}
P_{G,s} \myeq & \sum\limits_{u \in V} \bar{z}_u^2 = \bar{z}^T  \bar{z} &  \text{  [Polarization]}
\end{array}
\end{align}
\end{tcolorbox}

\noindent 
We now introduce the Polarization-Disagreement index, i.e., the objective we care about.

\begin{tcolorbox}
\vspace{-.5em}
\begin{align}
\label{eq:dfn_disagree}
\begin{array}{c@{}cr}
\mathcal{I}_{G,s} \myeq &\  P_{G,s} + D_{G,s}  &  \text{  [Polarization-Disagreement index] }
\end{array}
\end{align}
\end{tcolorbox}

\noindent Two useful propositions and an observation follow. 

\begin{proposition}  
\label{prop1}
The disagreement $D_{G,s}$ satisfies the equation: 
$$  D_{G,s} = \sum\limits_{(u,v) \in E} w_{uv} (\bar{z}_u - \bar{z}_v)^2.$$ 
\end{proposition} 

\begin{proof} 
Let $\mu \myeq \frac{z^{*T}  \vec{1}}{n} $ be the mean of $z^*$, so the $i$-th coordinate of $\bar{z}$ is $\bar{z}_i = z_i^*-\mu$. Now, observe that
$ w_{uv} (\bar{z}_u - \bar{z}_v)^2 =  w_{uv} ( z^*_u-\mu - z^*_v+\mu)^2 = d(u,v)$.  The result follows by summing over all edges. 
\end{proof}  

\begin{observation} \label{disagreementQuadratic}
The disagreement $D_{G,s}$ is a quadratic form. Specifically, $D_{G,s} = {z^*}^T L z^*$, and by Proposition~\ref{prop1}, $ D_{G,s} = \bar{z}^T L \bar{z}$. 
\end{observation}

\begin{proposition}  
\label{prop2}
Let $\bar{s} = s - \frac{s^T \vec{1}}{n} \vec{1}$ be the mean-centered innate opinion vector. Then, $\bar{z} = (I+L)^{-1} \bar{s}$. 
\end{proposition}

\begin{proof} 
For any graph, $L \vec{1}=0$. 
Therefore $(I+L)\vec{1} = \vec{1}$, or equivalently 
$\vec{1} = (I+L)^{-1} \vec{1}$. This implies that ${z^*}^T\vec{1} = s^T (I+L)^{-1} \vec{1} = s^T \vec{1}$. By these facts, we obtain that 
\begin{align*}
\bar{z} &= z^* - \frac{ {z^*}^T\vec{1}}{n} \vec{1} = (I+L)^{-1} s -\frac{{z^*}^T \vec{1}}{n} \vec{1} =  (I+L)^{-1} s - \frac{s^T \vec{1}}{n} \vec{1} \\ &=    (I+L)^{-1} \big( s- \frac{\vec{1}^T s}{n} \vec{1}\big) = (I+L)^{-1} \bar{s}.
\end{align*}
\end{proof}

\noindent  Proposition~\ref{prop2}    provides an alternative way of computing $\bar{z}$. We can either find it in the obvious way, i.e., find $z^*$ and then center it around zero. Alternatively, we can first center $s$ around zero, and then obtain $\bar{z}$ as the  Friedkin-Johnsen equilibrium when the innate opinion vector is $\bar{s}$.   

\spara{Trade-off between polarization and disagreement.} Before proceeding into any technical details, we show a simple example that illustrates the trade-off between polarization and disagreement. Suppose there are three agents, two of which have opinion 0, and one has opinion 1 on a certain topic, i.e., $s=[0,0,1]$. We wish to recommend one link with weight $1$ between these three agents.  The recommendation that agrees with human confirmation bias, and therefore does not cause any dissatisfaction to the three agents, is the edge between nodes $1$ and $2$.  The equilibrium opinion vector is the same as $s$, i.e., $z^*=[0,0,1]$. The total disagreement is 0, and the polarization is equal to $(-1/3)^2+(-1/3)^2+(2/3)^2= 0.667$.  An alternative choice is to recommend the edge between nodes $1,3$. The equilibrium now is $[1/3,0,2/3]$.  The total polarization is equal to $0^2+(-1/3)^2+(1/3)^2=0.222$. 
The total disagreement is $(2/3-1/3)^2=0.111$.  Therefore, the second recommendation resulted in a better outcome with respect to the sum of polarization and disagreement. By symmetry, edge $(2,3)$ has the same effect as edge $(1,3)$. The results are summarized in Table~\ref{tab:contro}.

\begin{table}[h]
\centering
\begin{tabular}{|c|cc|c|} \hline 
  Recommended link         &  $P_{G,s}$ &  $D_{G,s}$ &  $\mathcal{I}_{G,s}$ \\ \hline
 $(1,2)$     &  0.667   & 0  & 0.667  \\ 
 $(1,3)$     &  0.111     & 0.222  & 0.333  \\  
 $(2,3)$     &  0.111     & 0.222  & 0.333  \\  \hline
\end{tabular}
\vspace{.5em}
\caption{\label{tab:contro}Trade-off between polarization and disagreement for three agents with innate opinions $s_1=s_2=0,s_3=1$. For details, see Section~\ref{subsec:contro}.}
\end{table}
 \vspace{-2em}
\subsection{Optimizing over Graph Topologies  } 
\label{subsec:topologies}
Using the definitions of Section~\ref{subsec:contro}, we can now formulate Problem~\ref{prob1}  from Section~\ref{sec:intro} mathematically. The  objective is to minimize the sum of two terms,   polarization and   disagreement.  
  Here, $\mathcal{L}$ is the set of valid combinatorial Laplacians of connected graphs.   Observe that the trace of the Laplacian is 
equal to twice the total edge weight of the corresponding graph. 
 
\begin{tcolorbox}
\vspace{-.5em}
\begin{align}
\label{eq:p1}
  \text{ \underline{ Problem ~1 }  }
\begin{array}{ll@{}ll}
\text{min}_{L \in \field{R}^{n\times n}}   & \bar{z}^{T} \bar{z}  + \bar{z}^T  L \bar{z} &   \\
\text{subject to} & L \in \mathcal{L}   & & \\ 
						&\Tr(L) = 2m & & \\
\end{array}
\end{align}
\end{tcolorbox}

Note that by \eqref{eq:dfn_contr} and Observation \ref{disagreementQuadratic}, $\bar{z}^{T} \bar{z}  + \bar{z}^T  L \bar{z} = D_{G,s} + P_{G,s}$, where $G$ is the weighted graph corresponding to the Laplacian $L$. Thus \eqref{eq:p1} is equivalent to minimizing the polarization-disagreement index $\mathcal{I}_{G,s}$ over all graphs $G$ with total edge weight $m$.

\begin{theorem} 
\label{objconv}
The objective  $\bar{z}^T \bar{z}+\bar{z}^TL \bar{z}$ is a {\em convex} function of the edge weights in the graph $G$ corresponding to the Laplacian $L$. 
\end{theorem}

\begin{proof} Using Proposition \ref{prop2}
we rewrite the objective as: 

\begin{align*}
\bar{z}^T \bar{z}+\bar{z}^TL \bar{z} &= \bar{s}^T (I+L)^{-1} (I+L)^{-1} \bar{s} + \bar{s}^T (I+L)^{-1} L (I+L)^{-1} \bar{s} \\
&=  \bar{s}^T(I+L)^{-1} (I+L) (I+L)^{-1}\bar{s} = \bar{s}^T(I+L)^{-1}\bar{s}. 
\end{align*}

\noindent It is known that the function  $f(L) = (I+L)^{-1}$ is matrix-convex  when $L \in \mathcal{L}$ and hence positive semidefinite \cite{nordstrom2011convexity}. That is, for any $\lambda \in (0,1)$,
$$ \lambda Z_1^{-1}+ (1-\lambda) Z_2^{-1}   \succeq (\lambda Z_1 + (1-\lambda)Z_2)^{-1}.$$
This gives that $x^T (I+L)^{-1} x$ is convex for all vectors $x \in \field{R}^n$. 
\end{proof} 

\noindent Additionally, the set of Laplacians $ \mathcal{L}$ which we optimize over is convex (standard fact). 

\begin{claim}
\label{c1} 
The set $ \mathcal{L} \myeq  \{L^{n \times n}: L\text{~Laplacian}, \Tr(L) = 2m \}$ is convex. 
\end{claim}

By Theorem~\ref{objconv} and Claim~\ref{c1} we obtain that Problem 1 is solvable in polynomial time. One may use gradient descent or second order methods, see \cite{boyd2004convex}.  In the following, we show how to compute the gradient in a closed form in order to perform gradient descent more efficiently (compared to relying on numerical approximations of it). 

\spara{Gradient.} 
Let $N= {n \choose 2}$, and $e_i$ be the $i$-th standard basis vector in $\field{R}^N$. We can write $L = B \diag(w) B^T$ where $B\in\mathbb{R}^{n\times N}$ is the sign oriented incidence matrix and $w \in \mathbb{R}^N$ is the vector of edge weights for the graph $G$ corresponding  to $L$. Let $b_i$ be the $i$-th column of $B$, $i=1,\ldots,N$.  
 Observe that the  graph Laplacian corresponding to the weight vector $w +\epsilon e_i,$ is $L+\epsilon b_i b_i^T$.  Hence, if we perturb the $i$-th coordinate of any weight vector by $\epsilon$, i.e., to $w +\epsilon e_i$ we obtain the Laplacian  $L+\epsilon b_i b_i^T$.
 The Sherman-Morrison formula for the matrix pseudo-inverse  yields:
\begin{align*}
(I+L+\epsilon b_ib_i^T)^+ = (I+L)^{-1} -\epsilon \frac{(I+L)^{-1} b_ib_i^T(I+L)^{-1} }{1+\epsilon b_i^T(I+L)^{-1}b_i},
\end{align*} 
and hence, thinking of $(I+L)^{-1}$ as a matrix-valued function of $w$,
\begin{align*}
\frac{\partial (I+L)^{-1}}{\partial w_i} &= \\ 
\lim_{\epsilon \rightarrow 0} \frac{1}{\epsilon}\left[(I+L)^{-1} -\epsilon \frac{(I+L)^{-1} b_ib_i^T (I+L)^{-1}}{1+\epsilon b_i^T(I+L)^{-1}b_i} - (I+L)^{-1} \right] &= \\ 
 -(I+L)^{-1}b_ib_i^T (I+L)^{-1}.&
\end{align*}

Therefore, by  linearity, and the fact that the objective $\bar z^T \bar z + \bar z^T  L \bar z = \bar s^T (I+L)^{-1} \bar s$ as shown in Theorem \ref{objconv},
$$\frac{\partial s^T(I+L)^{-1}s}{\partial w_i}= -s^T(I+L)^{-1}b_ib_i^T (I+L)^{-1}s.$$

\spara{Non-convexity.} Perhaps surprisingly,   a slightly more general form of our objective, where one of the two terms is multiplied by any factor $\rho \geq 0$ (i.e., polarization and disagreement are weighted differently), is not convex! 

\begin{theorem} 
\label{nonconvexobjective}
Let $\rho \bar{z}^T \bar{z} + \bar{z}^TL \bar{z}, \rho\geq 0$ be our objective.  For $\rho=0$, the objective is a non-convex function of the edge weights. 
\end{theorem}   

\begin{proof}
By Propositions~\ref{prop1} and \ref{prop2} we obtain that 

\begin{align*}
z^{*T} L z^* &= \bar{z}^T L \bar{z} = \bar{s}^T (I+L)^{-1}L (I+L)^{-1} \bar{s}.
 \end{align*}  
 
 \noindent  To prove that this is not convex, it suffices to prove that $f(L) = (I+L)^{-1}L(I+L)^{-1}$ is  non-convex.
Let $L_1,L_2 \in \mathcal{L}$. Assume for the sake of contradiction that $f(L) = (I+L)^{-1}L(I+L)^{-1}$ is  convex. Then,  the following follows by convexity for any  $\lambda \in (0,1)$,

\begin{align*}
\lambda f(L_1) + (1-\lambda) f(L_2)  \succeq f(\lambda L_1 + (1-\lambda) L_2) \rightarrow \\  
 \lambda (I+L_1)^{-1}L_1(I+L_1)^{-1}+ (1-\lambda) (I+L_2)^{-1}L_2(I+L_2)^{-1}  \succeq \\
  (I+\lambda L_1+(1-\lambda)L_2)^{-1}(\lambda L_1+(1-\lambda)L_2) (I+\lambda L_1+(1-\lambda)L_2)^{-1}.
\end{align*}
 
Set $\lambda= 0.5$,  and consider the following two Laplacian matrices corresponding to two paths on 3 nodes.  

\vspace{2mm}
\begin{center} 
$L_1=\left[\begin{array}{ccc} 
     1  &  -1  &   0 \\ 
    -1  &   2  & -1 \\
     0  &  -1  &    1
\end{array}\right]$,
$L_2 = \left[\begin{array}{ccc} 
     2  &  -1  &  -1 \\ 
    -1  &   1  &   0\\
    -1  &   0  &   1 
\end{array}\right].$
\end{center}
\vspace{2mm}

\noindent It is easy to verify  numerically that  the minimum eigenvalue of 
$\lambda (I+L_1)^{-1}L_1(I+L_1)^{-1}  + (1-\lambda) (I+L_2)^{-1}L_2(I+L_2)^{-1} - 
  (I+\lambda L_1+(1-\lambda)L_2)^{-1}(\lambda L_1+(1-\lambda)L_2) (I+\lambda L_1+(1-\lambda)L_2)^{-1} $ is negative. This 
contradicts the convexity assumption of $f(L)$.  
\end{proof} 

\noindent We note that it is easy to construct more counterexamples for other $\rho$ values (details omitted).  

Finally, we prove that there always exists a graph with ${O}(\frac{n}{\epsilon^2})$ edges that lies within a multiplicative $(1 \pm \epsilon +O(\epsilon^2))$ factor of the optimal -- that is, approximately  minimizing the polarization-disagreement index $\mathcal{I}_{G,s}$ does not require a dense graph $G$.   Our proof relies on the seminal work on spectral sparsification \cite{spielman2011graph,spielman2011spectral,batson2012twice,spielman2014nearly}. 
The next theorem proves  a sub-optimal result in terms of the number of edges in $G$, but uses the sparsification algorithm due to Spielman and Srivastava based effective resistances, which has the advantages of being fast and easy to implement \cite{spielman2011graph}. 

\begin{theorem}
\label{sparsifiers}
There always exists a graph $G$ with   $O(n \log n/\epsilon^2)$ edges that achieves polarization-disagreement index $\mathcal{I}_{G,s}$ within  a  multiplicative $(1 +\epsilon +O(\epsilon^2))$ factor of optimal for Problem 1. 
\end{theorem}

\begin{proof}
Let $L$ be the Laplacian that minimizes  the sum of polarization and disagreement subject to total edge weight constraint $m$, i.e, the solution to \eqref{eq:p1}. By applying the Spielman-Srivastava algorithm to $L$ with error parameter $\epsilon'$, we obtain a spectral sparsifier whose combinatorial Laplacian $\tilde{L}$ satisfies for any $x \in \mathbb{R}^n$:
\begin{align}\label{eq:sparsifier}
(1-\epsilon') x^T L x \leq x^T \tilde{L} x \leq (1+\epsilon') x^T Lx.
\end{align}  

\noindent This in turn implies, 
\begin{align*}
(1-\epsilon') x^T (L+I) x &\leq x^T (\tilde{L}+I) x \leq (1+\epsilon') x^T (L+I) x \rightarrow \\ 
\frac{1}{1+\epsilon'} x^T (L+I)^{-1} x &\leq x^T (\tilde{L}+I)^{-1} x \leq \frac{1}{1-\epsilon'} x^T (L+I)^{-1} x. 
\end{align*}

\noindent By setting $x =s$,   $OPT \myeq s^T  (L+I)^{-1} s$, and using the Taylor expansion for the fraction $\frac{1}{1-\epsilon'}=1+\epsilon'+O(\epsilon'^2)$ when $\epsilon'$ is small, we obtain  that  $s^T (\tilde{L}+I)^{-1} s \leq  \left (1+\epsilon'+O(\epsilon'^2)\right )\cdot OPT$. 

Note that we may not have $\Tr(\tilde L) = 2m$. However, by \eqref{eq:sparsifier}, $\Tr(\tilde L) \le (1+\epsilon') \Tr(L) = (1+\epsilon')2m$. Thus we can scale $\tilde L$ to have trace $ \Tr(\tilde L) = 2m$ and will still have  $s^T (\tilde{L}+I)^{-1} s \le \frac{1}{1+\epsilon'} \cdot \left (1+\epsilon'+O(\epsilon'^2)\right) \cdot OPT \le \left((1 + \epsilon + O(\epsilon^2)\right)\cdot OPT$ if we set $\epsilon' = \epsilon/2$.
Finally, notice that $s^T (\tilde{L}+I)^{-1} s \geq s^T (L+I)^{-1} s = OPT$ since $L$ is the optimizer of \eqref{eq:p1}. 
\end{proof}

As mentioned, we use the Spielman-Srivastava sparsification algorithm \cite{spielman2011graph} in our proof since we use it in our experiments. However, a follow-up result due to Batson-Spielman-Srivastava  \cite{batson2012twice} reduces the number of edges to linear in the number of nodes. Using the same mathematical arguments, but invoking \cite{batson2012twice} instead of  \cite{spielman2011graph} we obtain the following corollary of Theorem~\ref{sparsifiers}.

\begin{corollary} 
There always exists a graph with    $O(n/\epsilon^2)$ edges that achieves polarization-disagreement index $\mathcal{I}_{G,s}$ within  a  multiplicative $(1 +\epsilon +O(\epsilon^2))$ factor of the optimal for Problem 1. 
\end{corollary}

\subsection{Optimizing over Innate Opinions} 
\label{subsec:innate} 
We next give a mathematical formulation of Problem~\ref{prob3}. We show how to simplify this formulation to obtain a convex optimization program (specifically, an SDP) which can be solved in polynomial time using standard algorithms \cite{boyd2004convex}.

 Equation \eqref{eq:opts}  provides a straight-forward way to model our problem. We wish to minimize the sum of polarization and  disagreement subject to the structure implied by  the dynamics, and a total budget  $\alpha$ on the total change of  innate opinions.  The variable we optimize over is $ds$, the change in opinions. We restrict this change to decrease the innate opinions, i.e., $ds \leq \vec{0}$. 
 
  \begin{tcolorbox}
  \vspace{-.5em}
\begin{align}
\label{eq:opts}
\underline{Problem~\ref{prob3}}
\begin{array}{ll@{}ll}
\text{min}_{ds \in \field{R}^n} & \bar{z}^{T} \bar{z} + \bar{z}^{T}L \bar{z} &   \\
\text{subject to} &z^*  = (I+L)^{-1} (s+ds) & & \\ 
                       & \bar{z} = z^* - \frac{\vec{1}^T z^*}{n} \vec{1} & & \\
					    & \vec{1}^T ds \geq -\alpha & & \\
						&  ds \leq \vec{0} & & \\
						& s+ ds \geq \vec{0} & & \\  
\end{array}
\end{align}
\end{tcolorbox}

\hide{

 \begin{tcolorbox}
\begin{align}
\label{eq:opts2}
\underline{Problem~\ref{prob3}}
\begin{array}{ll@{}ll}
\text{min}_{ds \in \field{R}^n} & \bar{z}^{T} \bar{z} + \bar{z}^{T}L \bar{z} &   \\
\text{subject to} &z^*  = (I+L)^{-1} (s+ds) & & \\ 
                          & \bar{z} = z^* - \frac{1}{n}(\vec{1}^T z^*) \vec{1} & & \\
						& \vec{1}^T ds \leq \alpha & & \\
					    & \vec{1}^T ds \geq -\alpha & & \\
						& s+ ds \leq \vec{1} & & \\
						& s+ ds \geq \vec{0} & & \\  
\end{array}
\end{align}
\end{tcolorbox}
 
 }

 \noindent Our main result is the following proposition.   
\begin{proposition}
\label{prop3}    
The formulation of Problem 2 in~\eqref{eq:opts} is solvable in polynomial time. 
\end{proposition} 
 \begin{proof}
 We can simplify the objective of~\eqref{eq:opts} using our analysis from Section~\ref{subsec:topologies}. Specifically, our problem is equivalent to 
 \begin{align}
 \begin{array}{ll@{}ll}
\text{minimize}  & (s+ds)^T (I+L)^{-1} (s+ds) &   \\
\text{subject to} &  ds \leq \vec{0} & & \\
						& \vec{1}^T ds \geq -\alpha & & \\
						& s+ ds \geq \vec{0} & & \\  	
\end{array}
\end{align}
By expanding the above objective, we observe that it is a standard quadratic form $x^TQx+2b^Tx+c$ where $Q \myeq (I+L)^{-1}$ is symmetric positive semidefinite, $b = (I+L)^{-1} s$ and $c = s^T (I+L)^{-1} s$ (note that $b,c$ are fixed with respect to the variable $ds$). This objective is convex, and the set of constraints form a convex set, proving the Proposition.  
\end{proof}
We remark that the above convex formulation \eqref{eq:opts} can accommodate other types constraints, such as restricting $ds$ to be in a specific range of values, or allowing positive changes in the innate opinions. In our experimental section, we use the  basic formulation without additional constraints.

\section{Experimental results}
\label{sec:exp}
\subsection{Experimental Setup}
\label{subsec:setup} 
 
\spara{Datasets.} We use two datasets collected by De et al. \cite{de2014learning}. Specifically,  De et al.  collect text produced by users, and interactions between them \cite{de2014learning}. The text is mapped into opinions using NLP tools \cite{pennebaker2001linguistic}, and the interactions are used to create a network between users. There are many preprocessing details, which can be found in \cite{de2014learning}. Here, we provide a brief description of both datasets. 

{\em Twitter}  is a network with  $n=548$ nodes, and $m=$ 3,638 edges. Edges correspond to user interactions.  All nodes' opinions are available to us. This dataset was collected with the aim of analyzing the debate on Twitter about the Delhi legislative assembly elections of 2013. This was an irresolute event with three major parties winning roughly equal shares of the vote. The hashtags used to collect tweets were \#BJP, \#APP, \#Congress, \#Polls2013, tweeted over the period December 9th to 15th.

{\em Reddit} is a network with $n=556$ and $m=$ 8,969 edges. There is an edge between two users if there exist two subreddits (other than politics) that both users posted in during the given time period.  The topic of interest is politics.  

Aside from the above real-world datasets, we use a wide variety of synthetic data. Since most real-world networks have a skewed degree distributions, we focus on power law random networks generated using the Norros-Reittu model \cite{norros2006conditionally}, an important random graph model  that produces networks that mimic real-world networks in some respects \cite{aiello2000random,frieze2015introduction}. We also use skewed distributions for generating opinions. We use the {\em randht.m} 	file by Aaron Clauset \cite{clauset2009power} to generate opinions according to a power law with a given slope. We normalize opinions to the range $[0,1]$ by dividing by the maximum observed value (i.e., there is always a node with opinion 1).

\spara{Machine specs.} All experiments were run on a laptop with 1.7 GHz Intel Core i7 processor and 8GB of main memory.  

\spara{Code.} Our code was written in Matlab. Our code is publicly available at  \url{https://github.com/tsourolampis/polarization-disagreement}.

\subsection{Experimental Findings}
\label{subsec:findings}

  \begin{table}[t]
\centering \small
\begin{tabular}{r|cccc|cc|}
\multicolumn{1}{c}{} &  \multicolumn{4}{c}{} &  \multicolumn{2}{c}{Proposed method} \\
\cline{2-7}
\multicolumn{1}{c}{} &
\multicolumn{1}{|c}{ER(0.5)} &
\multicolumn{1}{c}{PL(2)} &
\multicolumn{1}{c}{PL(2.5)} &
\multicolumn{1}{c|}{PL(3)} &
\multicolumn{1}{c}{$L^*$} &
\multicolumn{1}{c|}{$\tilde{L}^*$-sparsified} \\ \cline{2-7}
$s \sim PL(1.5)$    & 14.38  & 16.10 &   22.06  &   53.05 & 11.60 & 11.60 \\ 
$s \sim PL(2)$      & 25.98 &  45.16 &    72.11  &   107.23 & 19.24 & 19.27 \\ 
$s \sim PL(2.5)$  & 94.87 & 103.62 &    121.21 & 166.38 & 85.55 & 85.56 \\ 
\cline{2-7}
\end{tabular}
\vspace{.3em}
\caption{\label{tab:learng} Polarization-disagreement minimization for random power law opinions vectors.  Average polarization-disagreement indices $\mathcal{I}_{G,s}$ over 5 experiments for Erd\"{o}s-R\'{e}nyi, random power law graphs, and the optimal graphs $L^*$ and $\tilde L^*$-sparsified output by our solutions to Problem 1.  Rows correspond to  generating innate opinions $s$ according to a power law distribution with slopes 1.5, 2, and 2.5 respectively.  
For details see Section~\ref{subsec:findings}.}
\end{table}

\begin{figure*}[!ht]
\centering
\begin{tabular}{@{}c@{}@{\ }c@{}@{}c@{}@{\ }c@{}}
\includegraphics[width=0.23\textwidth]{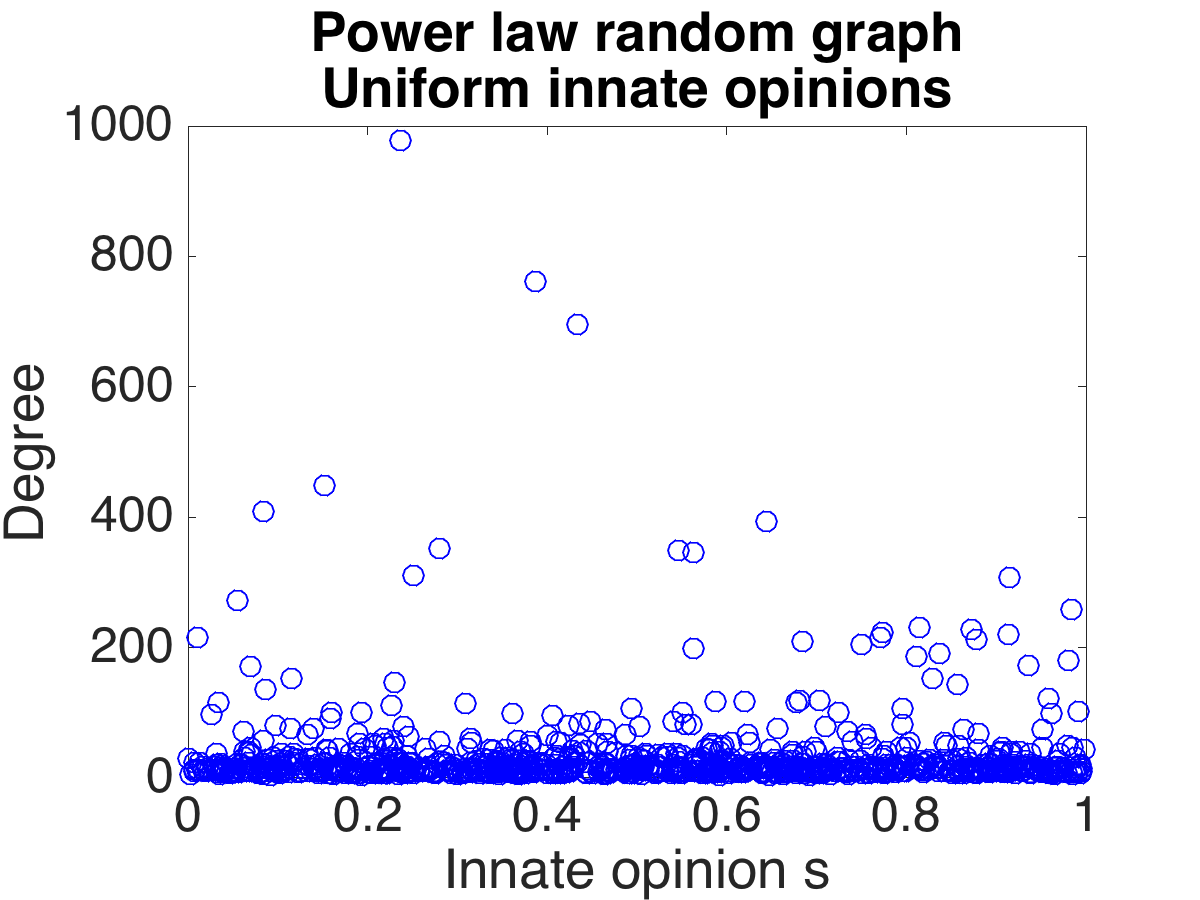}  &  \includegraphics[width=0.23\textwidth]{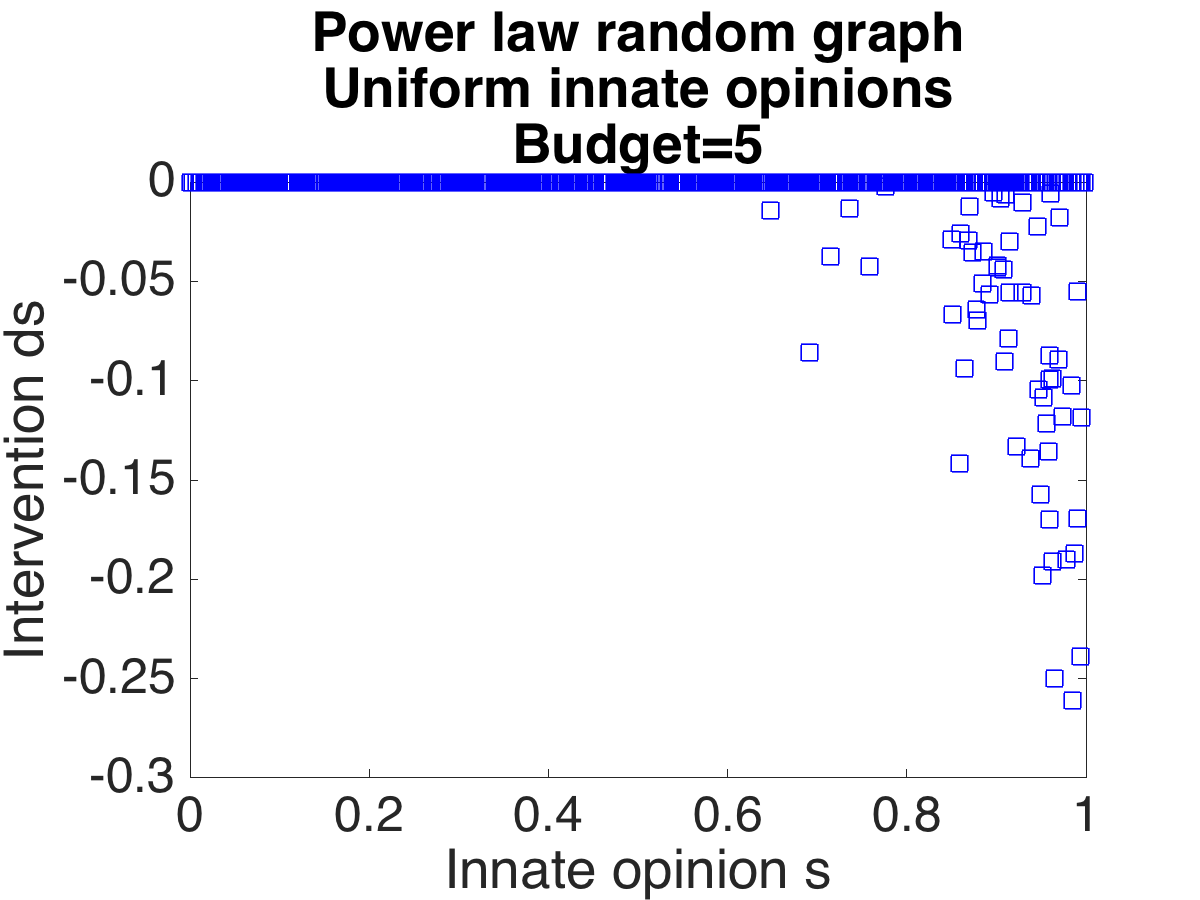} & \includegraphics[width=0.23\textwidth]{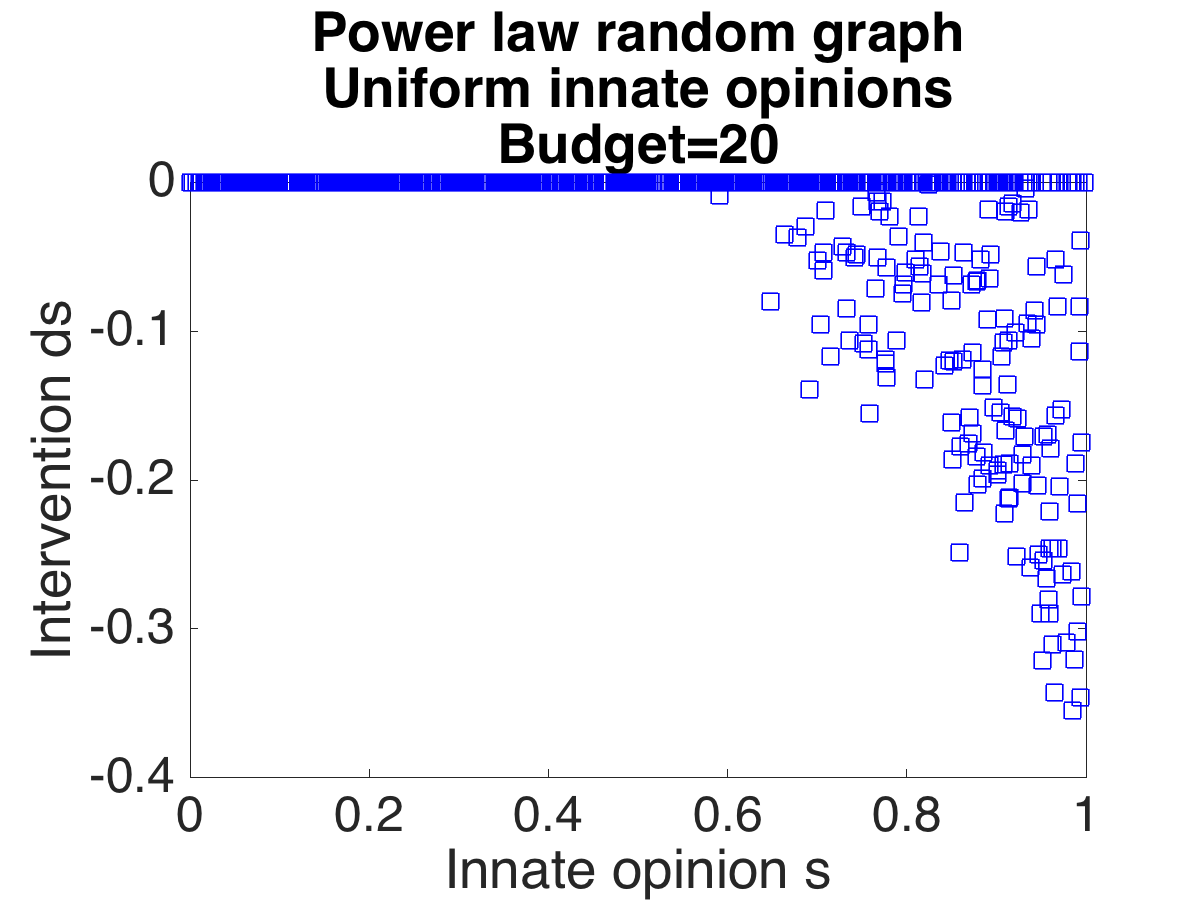}&  \includegraphics[width=0.23\textwidth]{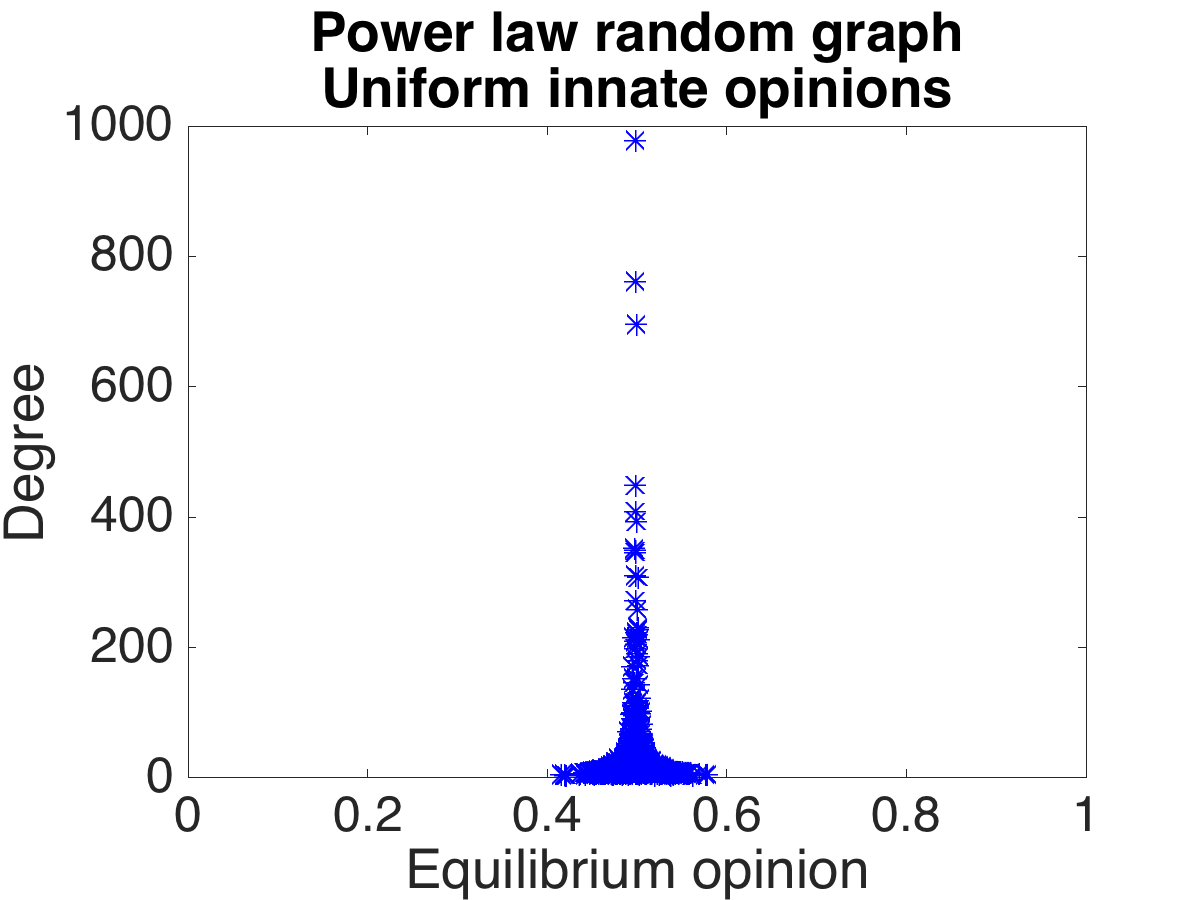}  \\
(a) &  (b) &  (c) &  (d) \\ 
\includegraphics[width=0.23\textwidth]{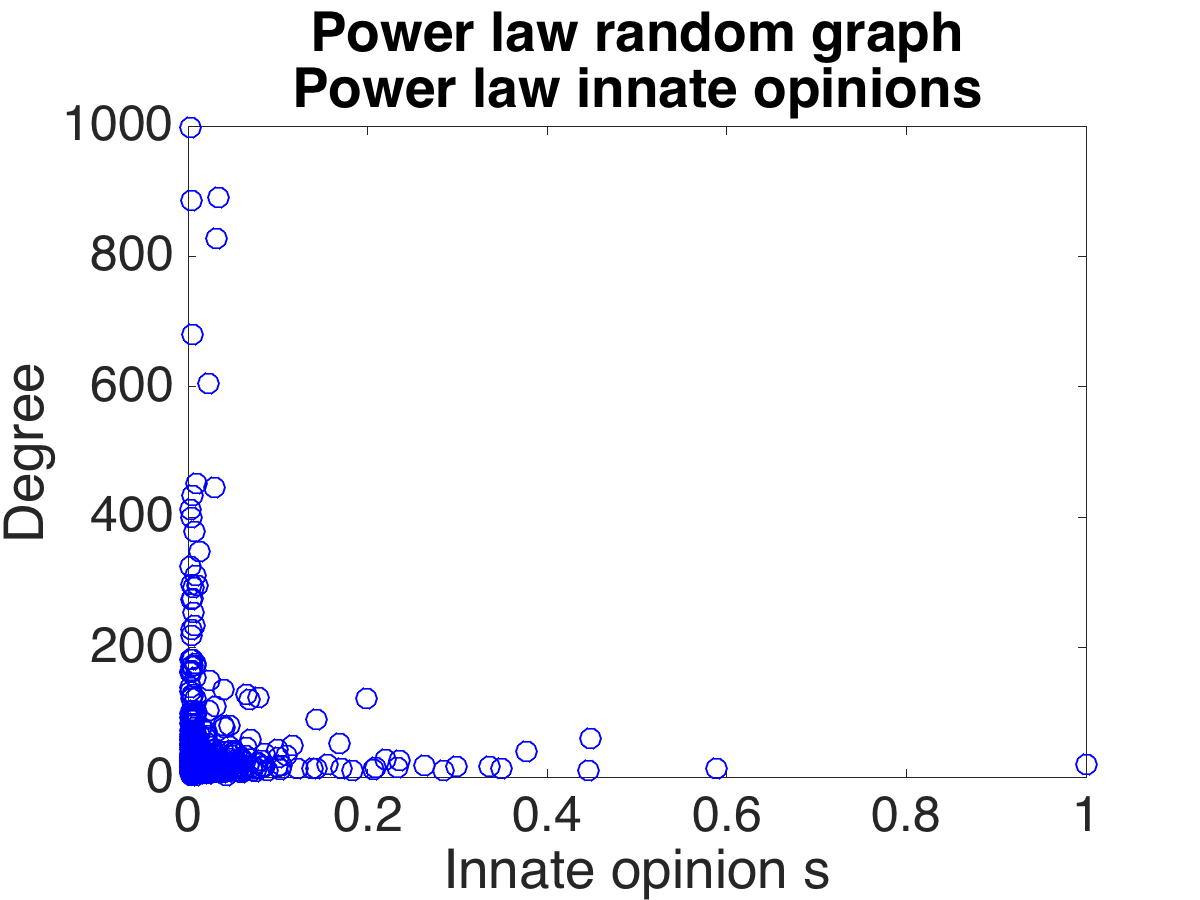}  &  \includegraphics[width=0.23\textwidth]{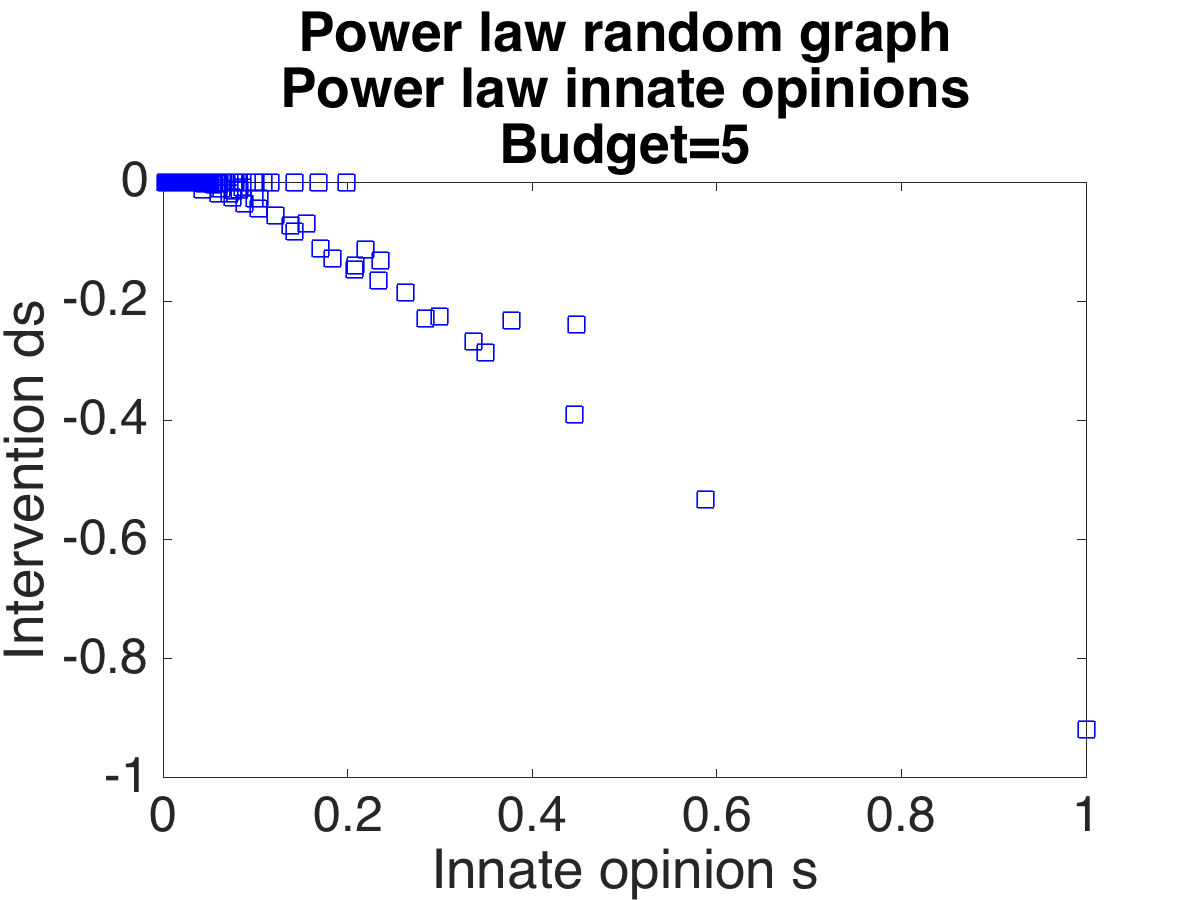} & \includegraphics[width=0.23\textwidth]{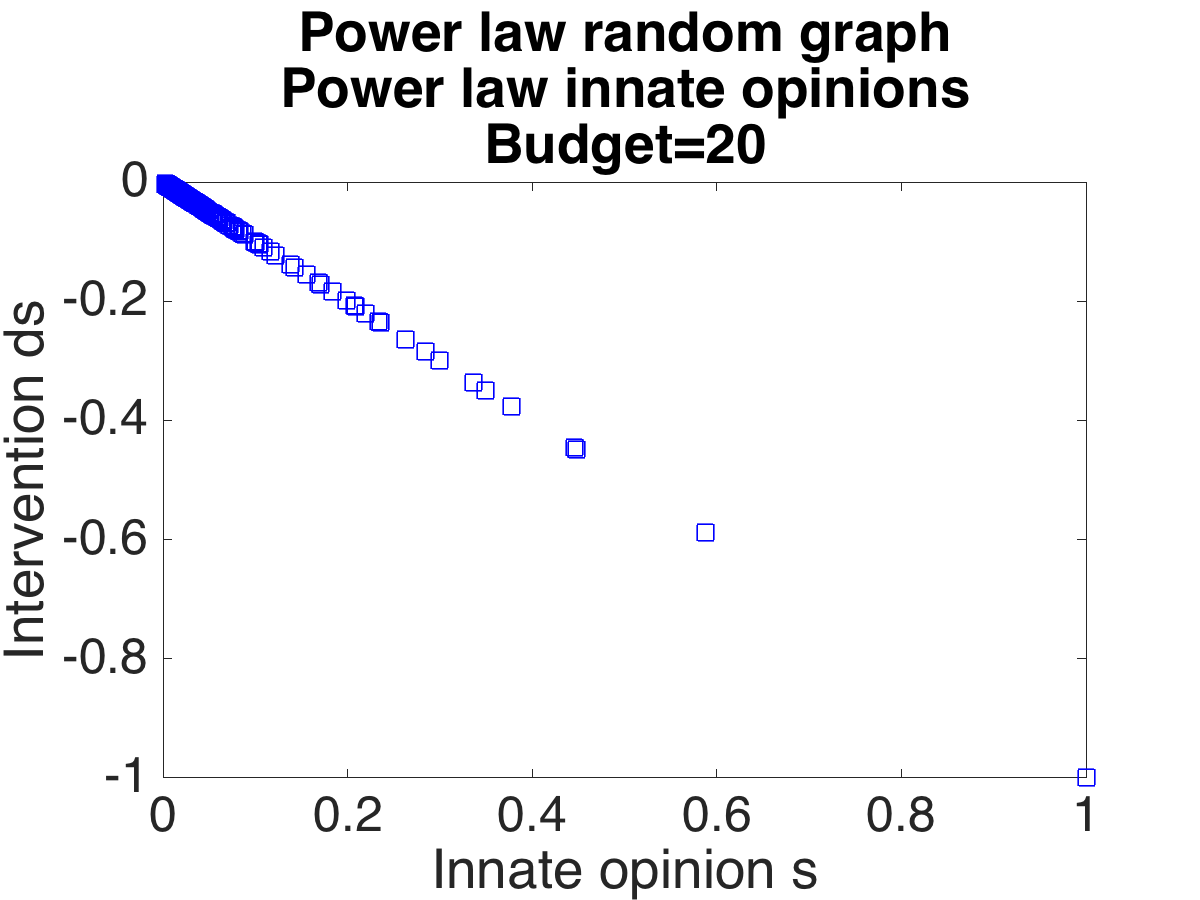}&  \includegraphics[width=0.23\textwidth]{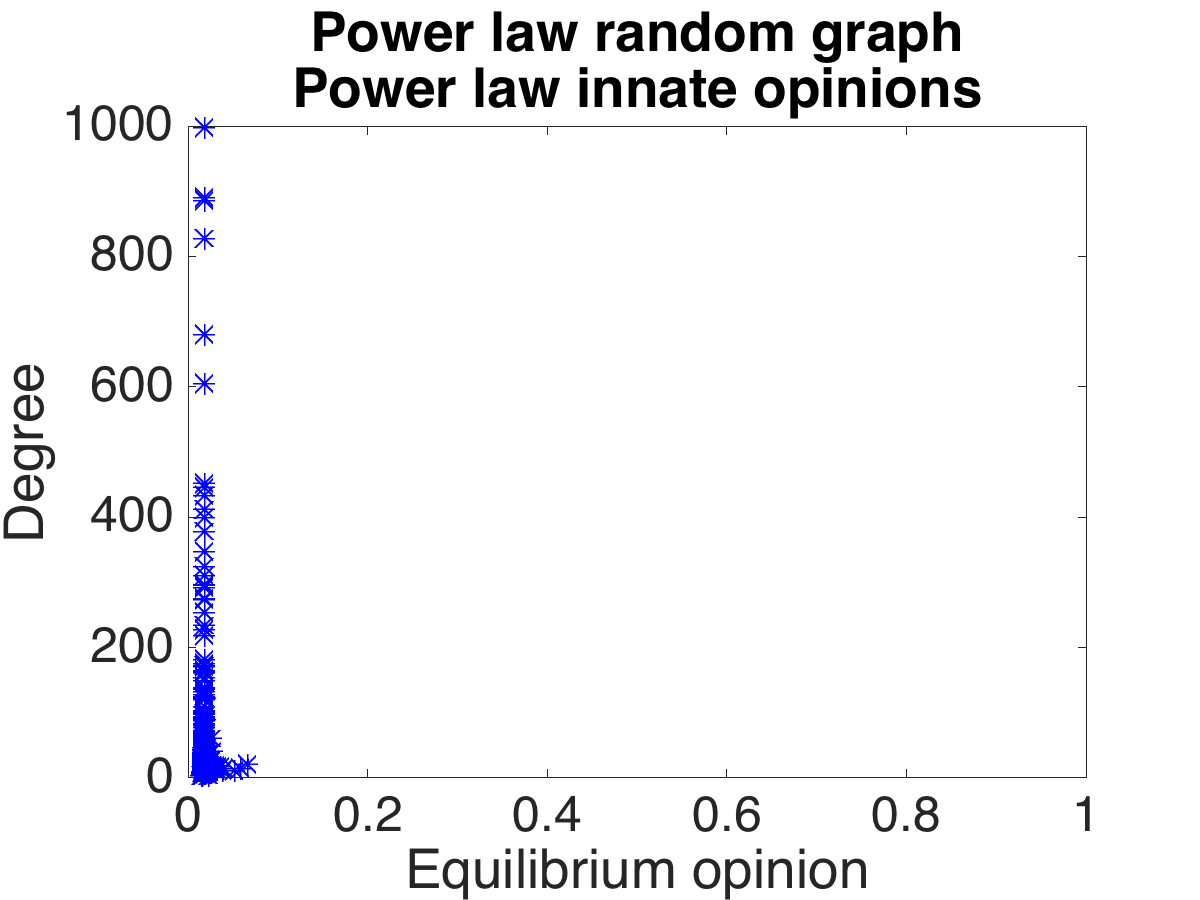}  \\
(e) & (f) & (g) & (h) \\ 
\includegraphics[width=0.23\textwidth]{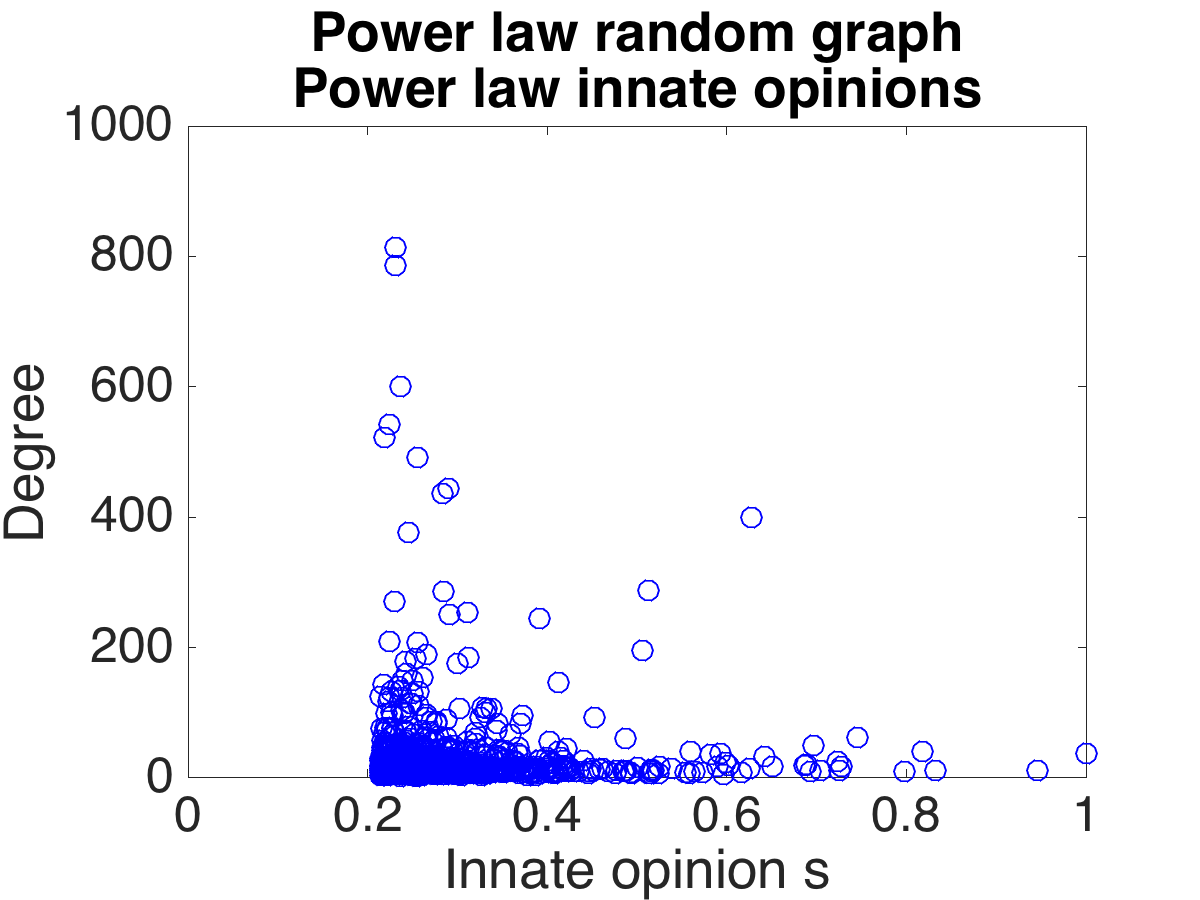}  &  \includegraphics[width=0.23\textwidth]{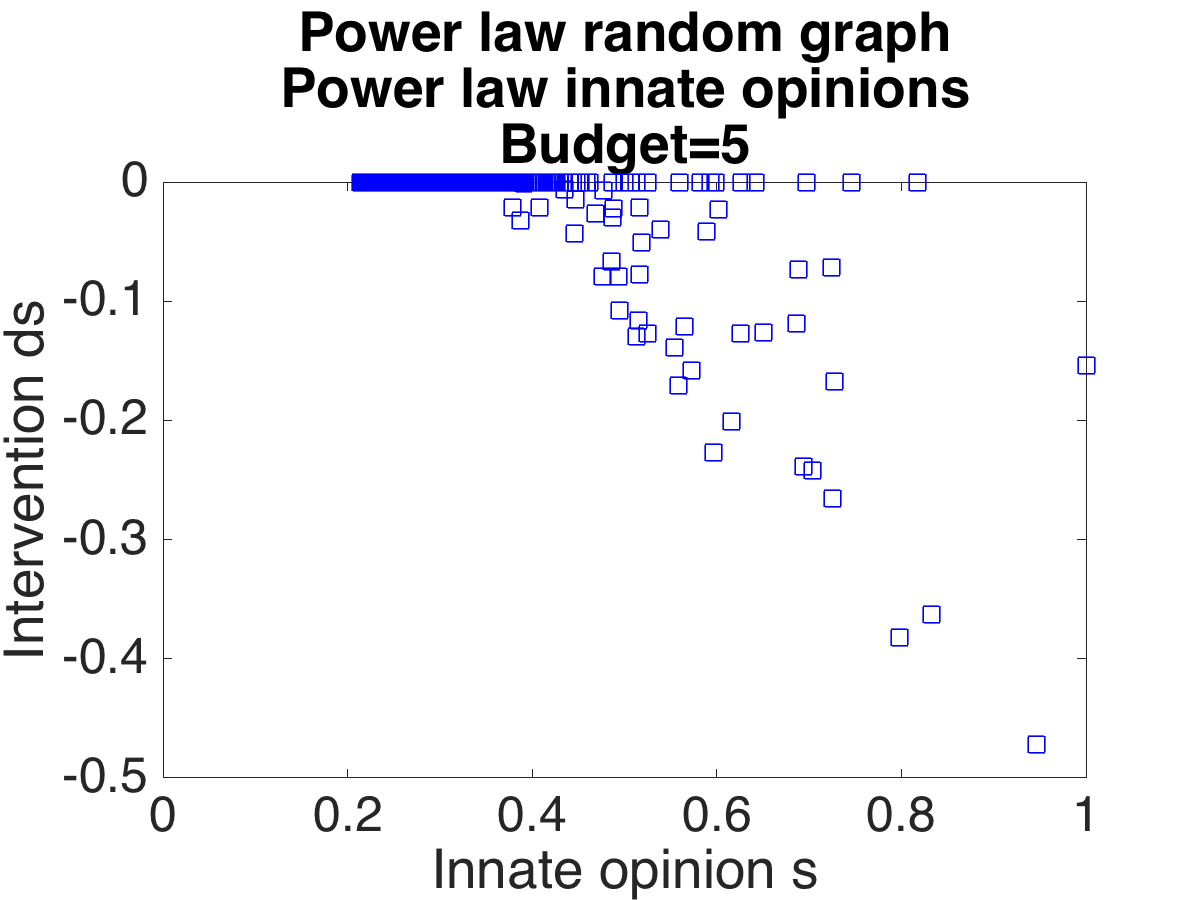} & \includegraphics[width=0.23\textwidth]{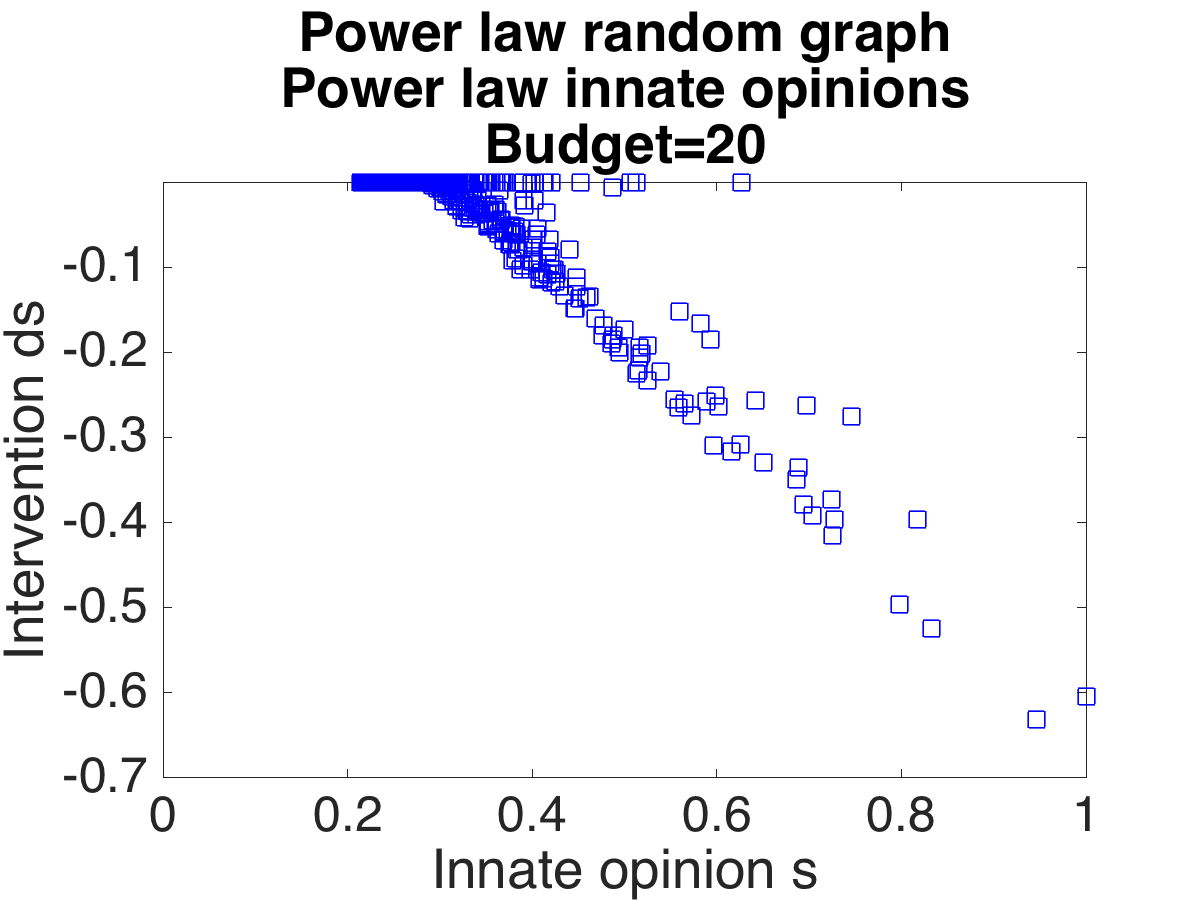}&  \includegraphics[width=0.23\textwidth]{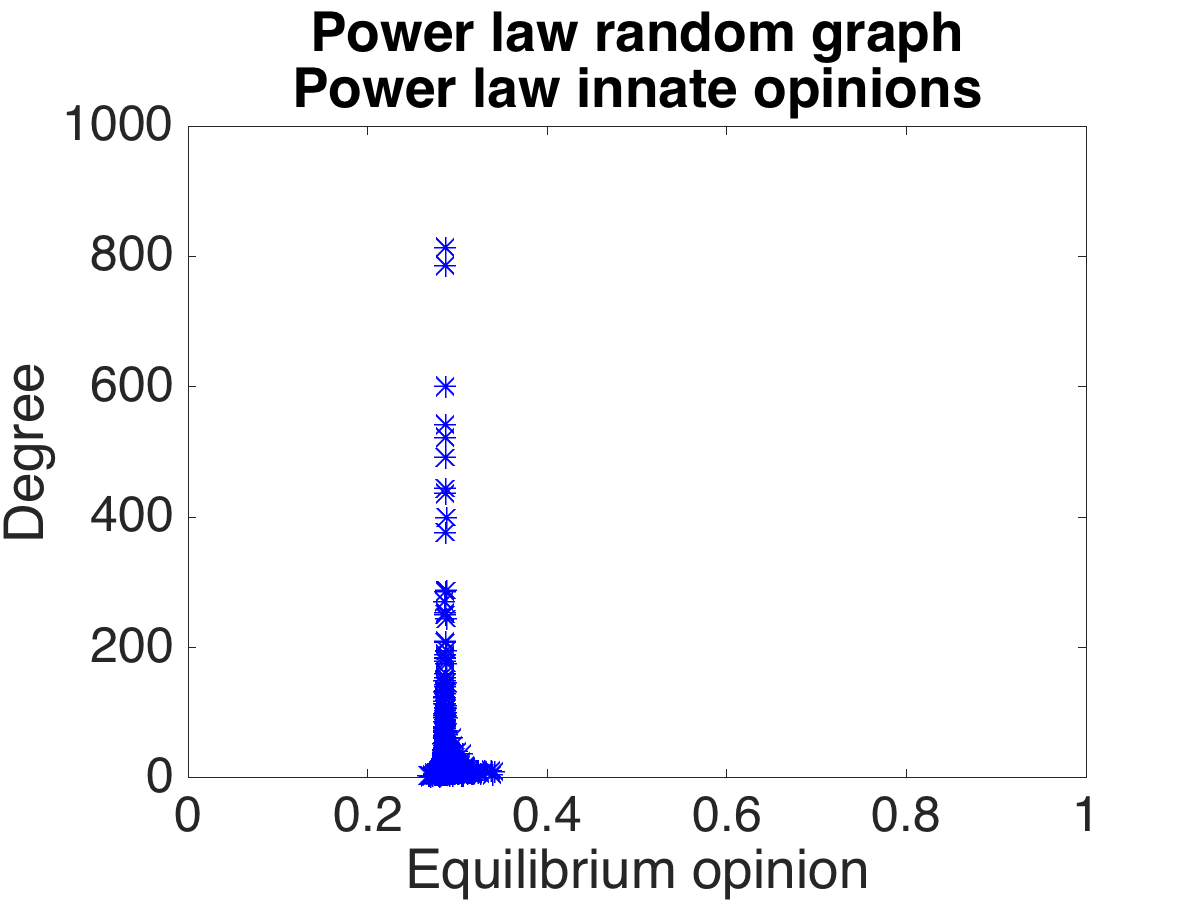}  \\
(i) & (j) & (k) & ($\ell$) \\  
\includegraphics[width=0.23\textwidth]{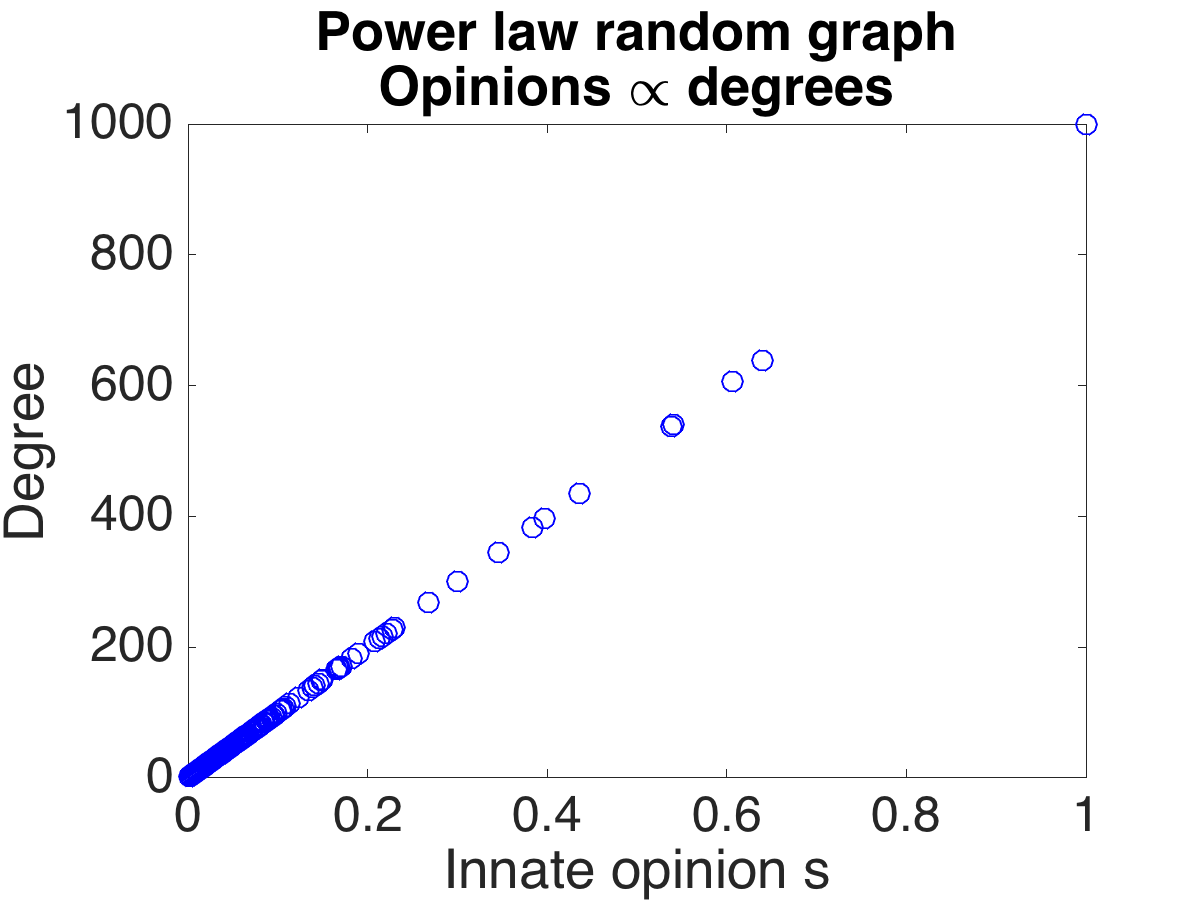}  &  \includegraphics[width=0.23\textwidth]{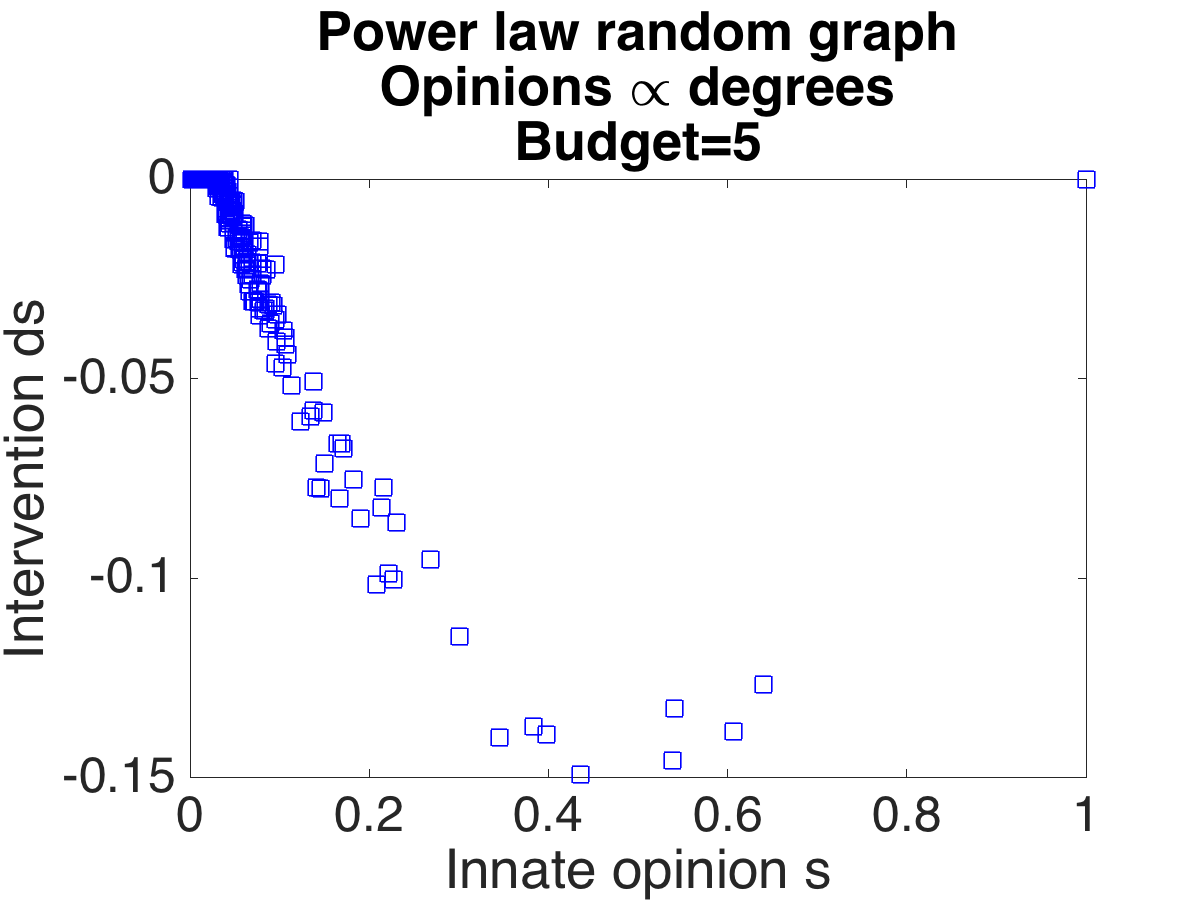} & \includegraphics[width=0.23\textwidth]{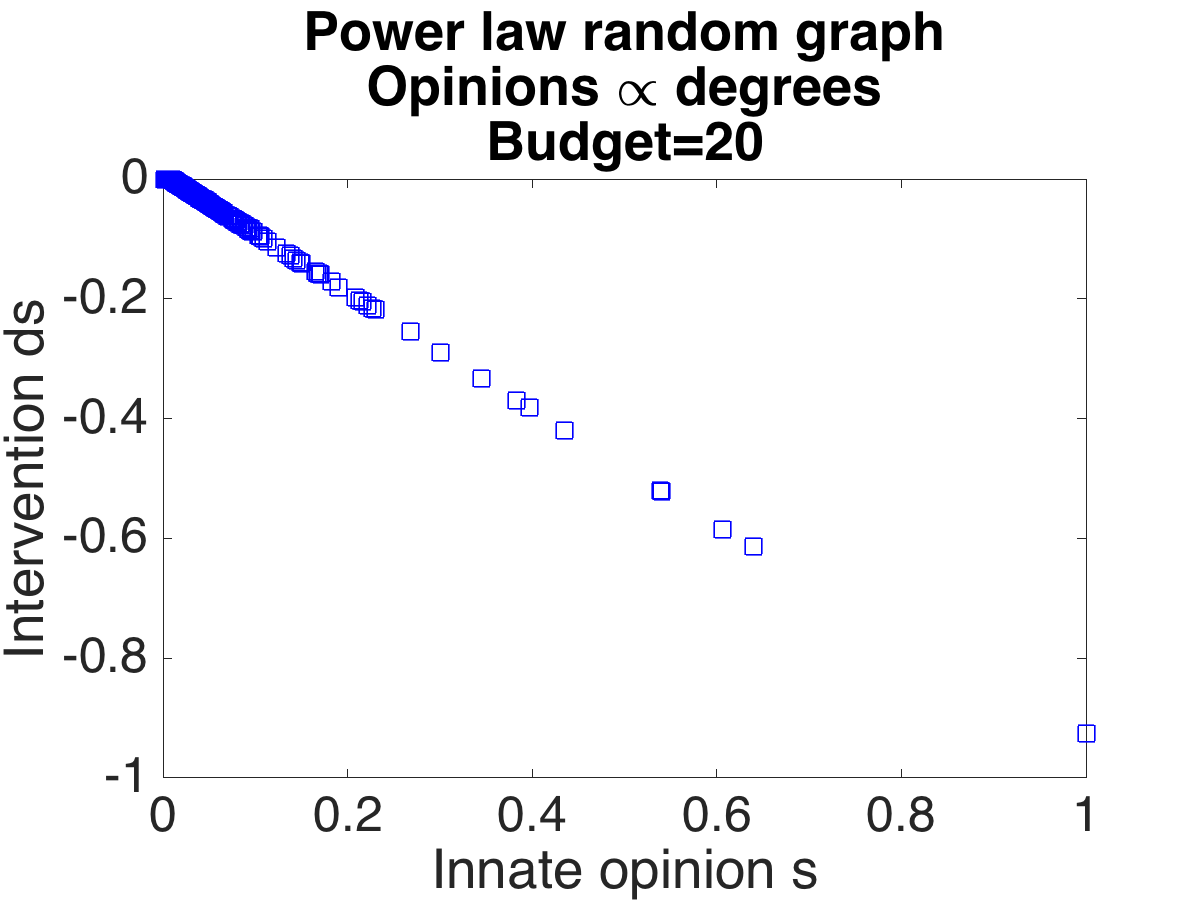}&  \includegraphics[width=0.23\textwidth]{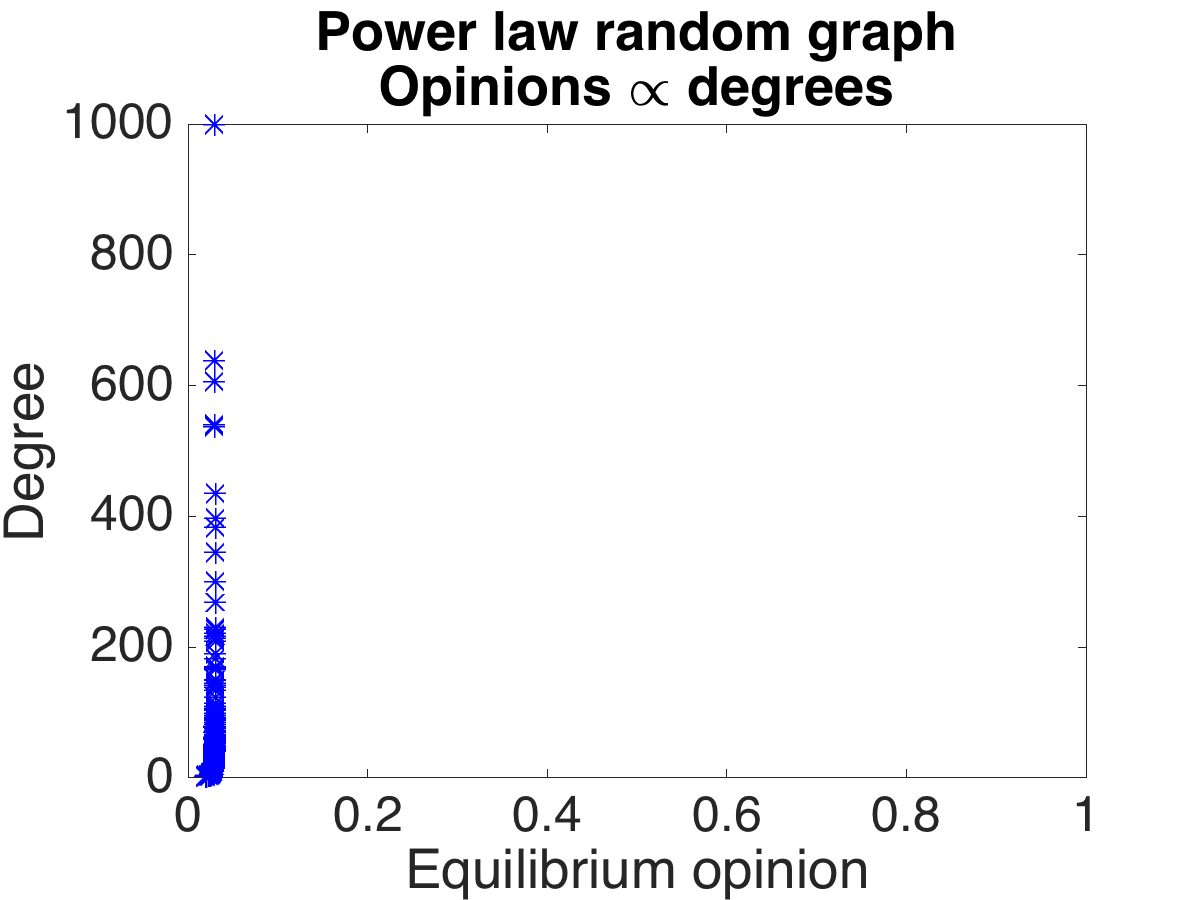}  \\
(m) & (n) & (o) & (p) \\  
\end{tabular}
\caption{\label{fig:synthetic1}  Optimizing over opinion interventions for varying innate opinions distributions on random power law graph topologies generated according to the Norros-Reittu model \cite{norros2006conditionally}. For details, see Section~\ref{subsec:findings}.}
\end{figure*}

\spara{Synthetic experiments.}   Table~\ref{tab:learng} shows our results on learning optimal graph topologies (i.e. solving Problem 1) for various innate opinion 
vectors averaged over 5 experiments. Each row corresponds to sampling a vector of 100 innate opinions according to the power law distribution with slopes 1.5, 2, and 2.5 respectively. Column $L^*$  shows the objective value achieved by our optimization algorithm, and the last column shows the $\tilde{L}^*$-sparsified the objective value after we sparsify $L^*$ using the Spielman-Srivastava algorithm (see Theorem \ref{sparsifiers}). We observe a negligible loss, but achieve a dramatic reduction in the number of edges. The optimal solution $L^*$ in all experiments corresponds to a weighted clique, i.e., there are 4950 non-zero edges.   The sparsified versions of the optimal solution have on average 761, 697, and 700 edges for the three different distributions on $s$ respectively.  

For comparison, in Table~\ref{tab:learng} we also show the polarization-disagreement index for several randomly generated graphs. The first column shows the results obtained by a random binomial graph, i.e., an Erd\"{o}s-R\'{e}nyi graph. We use $p=0.5$ as the underlying edge probability.  We also generate power law networks according to the Norros-Reittu model with slopes 2, 2.5, and 3 respectively.  In general, we observe that random binomial graphs achieve a close-to-optimal performance, in contrast to power law networks.  Understanding this observation using formal random graph theory \cite{frieze2015introduction} is an interesting direction for future research.  We also observe empirically that the polarization-disagreement index increases as the slope of the random network increases from 2 to 3.

Figure~\ref{fig:synthetic1} shows our experimental findings when we optimize over innate opinions. We generate power-law random graphs using the Norros-Reittu model with slope 2, a slope that lies close to the slopes of several real-world networks \cite{newman2005power}.  Similarly we generate random initial opinion vectors according to various distributions.   The first row (Figures~\ref{fig:synthetic1}(a),(b),(c),(d)) corresponds to uniform innate opinions. The second row (Figures~\ref{fig:synthetic1}(e),(f),(g),(h)) to power law opinions with slope 2, and the third row (Figures~\ref{fig:synthetic1}(i),(j),(k),($\ell$)) to slope 5. For the fourth row  (Figures~\ref{fig:synthetic1}(m),(n),(o),(p)) we set the opinion of each  node $v$ deterministically to $\frac{deg(v)}{\sum_u deg(u)}$.  This setting allows us to clearly see if the algorithm chooses to change the innate opinions of only the highest degree nodes, or if the graph topology forces the algorithm to choose other nodes (which is actually the case). 
 The first column (Figures~\ref{fig:synthetic1}(a),(e),(i),(m)) plots the degree of each node versus its innate opinion.   Figures~\ref{fig:synthetic1}(b)(f)(j),(n), and   Figures~\ref{fig:synthetic1}(c),(g),(k),(o) plot the optimal  change $ds$ for each node versus its innate opinion for budgets $\alpha=5$, and $\alpha=20$ respectively. We observe that the optimization algorithm tends to prefer to reduce high innate opinions. However we can also observe that graph topology plays a key role, as in certain cases it changes opinions of nodes with lower innate opinions even if it could choose higher innate opinions. This is obvious in all figures, but most striking in Figure~\ref{fig:synthetic1}(n); the algorithm does not modify the node with innate opinion 1 when the budget is 5, but only modifies other nodes.  
Similar remarks can be made for the rest of the figures.  This fact is more pronounced when $\alpha=5$ than when $\alpha=20$. since as the budget $\alpha$ increases, the algorithm can always reduce the values of the highest innate opinions, caring less and less about the graph topology.  Finally, Figures~\ref{fig:synthetic1}(d),(h), ($\ell$), (p) plot the  opinion vector at equilibrium after our algorithm's intervention (i.e, $z = (I+L)^{-1}(s+ds)$), when the budget is 20.  We observe that in all cases opinions get closer, but the actual value that they concentrate around depends on the distribution. For instance, in Figure~\ref{fig:synthetic1}(d) we observe concentration around 0.5 but in  Figure ~\ref{fig:synthetic1}(h) around 0. This should not be very surprising; the average of the uniform opinion distribution is 0.5, whereas when $s$ is distributed as a power law distribution, most of the values initially are close to 0.
The results presented in Figure~\ref{fig:synthetic1} are representative observations across many experiments.

\begin{table}[!htp]
\centering
\begin{tabular}{|cc|cc|}
\multicolumn{2}{c}{Twitter}  & \multicolumn{2}{c}{Reddit}  \\   \hline
$\mathcal{I}_{Twitter,s}$  & 199.84 & $\mathcal{I}_{Reddit,s}$  &  138.02 \\ 
\# Edges &  3,638   & \# Edges  &  8,969  \\ \hline 
$\mathcal{I}_{G^*,s}$ & 30.113 & $\mathcal{I}_{G^*,s}$ &  0.0022  \\ 
\# Edges & 120,000 & \# Edges &  103,050 \\  \hline 
$\mathcal{I}_{\tilde{G}^*,s}$ & 30.114 & $\mathcal{I}_{\tilde{G}^*,s}$ &  0.0022 \\
\# Edges & 3,455 &   \# Edges &  7,521   \\  \hline 
\end{tabular}
\vspace{.3em}
\caption{\label{tab:twitter} Optimizing over graph topologies. First row, shows the objective, and the number of edges for the Twitter and Reddit networks. Second row shows  the objective and the number of edges for the optimal solution. Third row shows the   objective and the number of edges for the sparsified optimal solution. For details, see Section~\ref{subsec:findings}. }
\end{table}

\spara{Real-world experiments.}   Table~\ref{tab:twitter} summarizes our experimental results for the Twitter and Reddit datasets. Given the Twitter network among the users, and their innate opinions, the polarization-disagreement index $\mathcal{I}_{G,s}$  is equal to 199.84.  
By optimizing over graph topologies, we find a graph $G^*$ with the same total weight 8,969 as the original graph that reduces the polarization-disagreement index to 30.113. However, this optimal graph is  dense; the total weight is spread over 120,000  edges. By sparsifying using effective resistances according to Spielman-Srivastava, we obtain a sparse graph using just 3,455 edges. The polarization-disagreement index  is equal to 30.114. Again, we observe that there is a negligible change in the objective value (on the order of $10^{-3}$), but there is a drastic reduction in the number of edges.   

The results for Reddit are impressive; the polarization-disagreement index  reduces from 199.84 to 0.0022.  Both the original graph and the optimal graph have the same total weight. The  weight of the optimal graph is spread over 103,050  edges, but by sparsifying we get a graph with 7,521 edges  with polarization-disagreement index  0.0022; the difference from the optimal value is $
<10^{-5}$.

\begin{figure*}[!ht]
\centering 
\begin{tabular}{cc}
\includegraphics[width=0.4\textwidth]{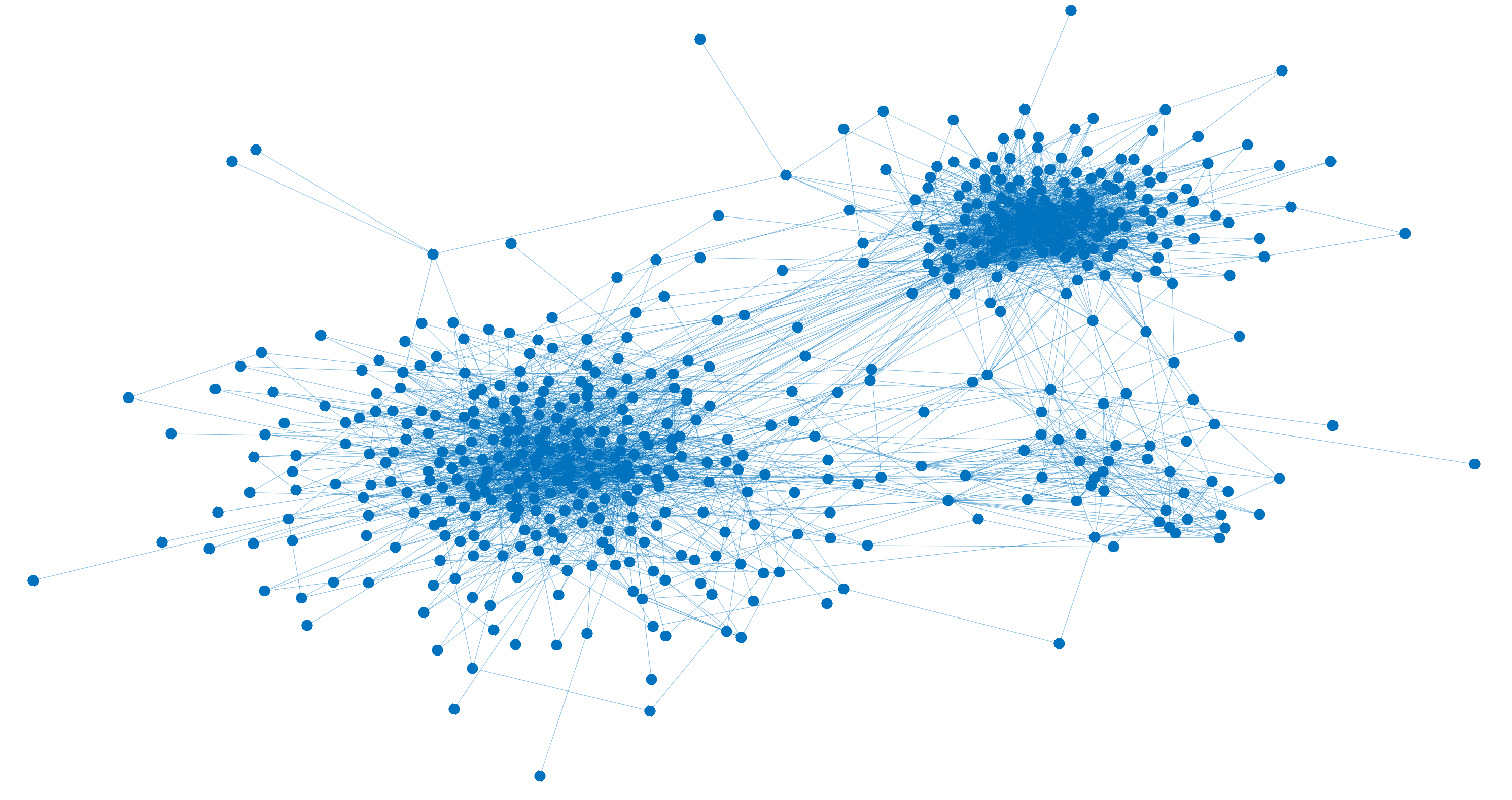} & \includegraphics[width=0.25\textwidth]{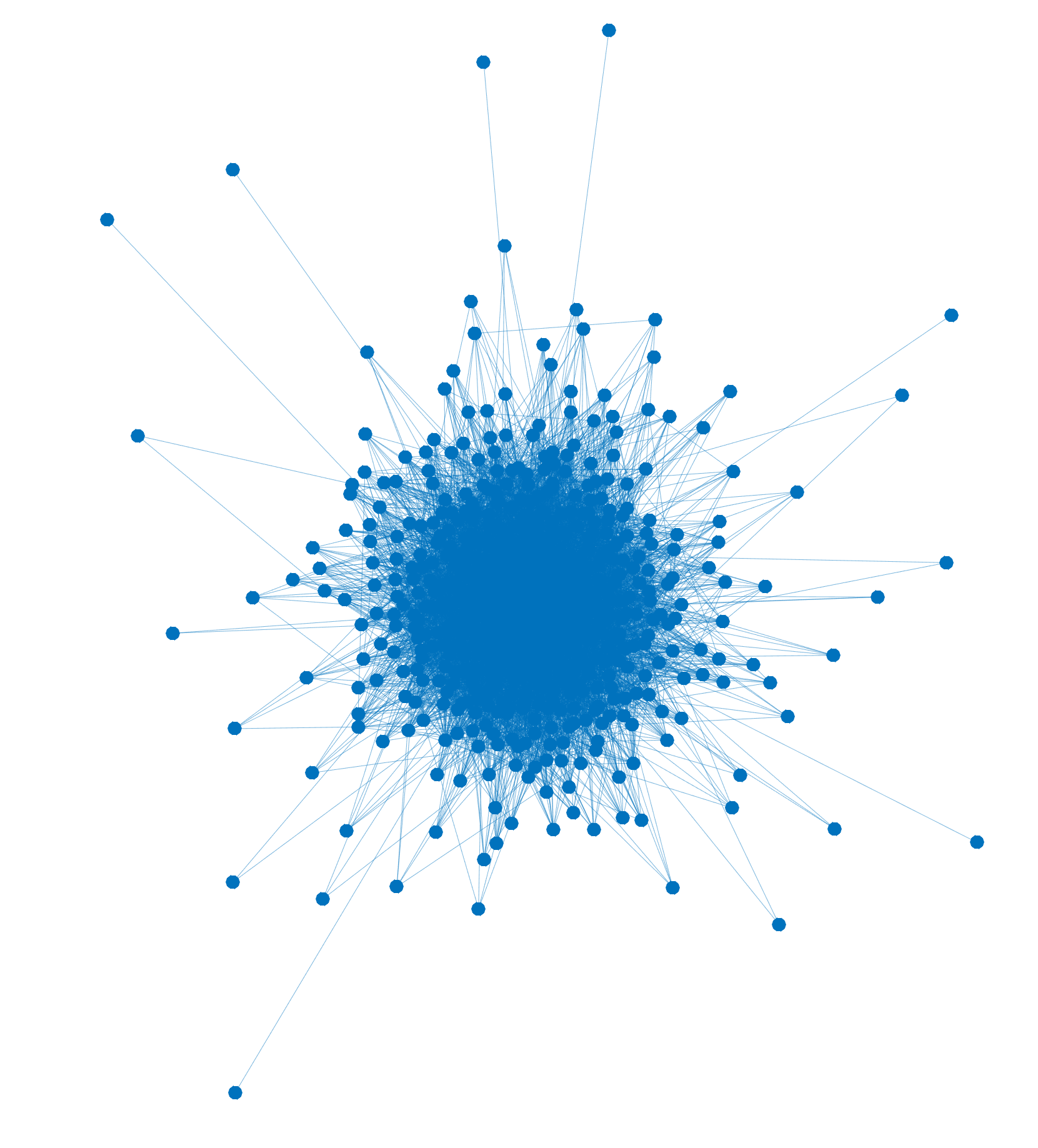} \\ 
(a) & (b) \\ 
\includegraphics[width=0.45\textwidth]{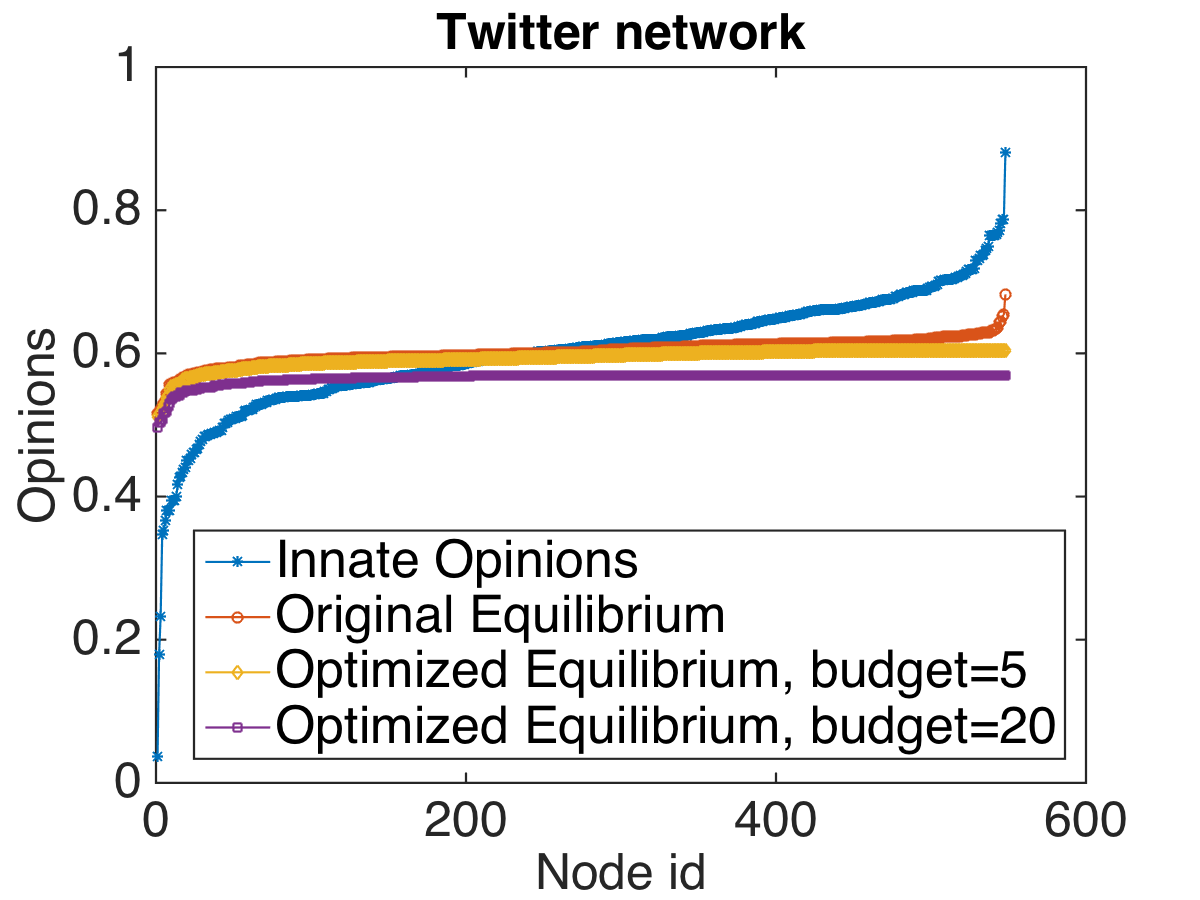} & \includegraphics[width=0.45\textwidth]{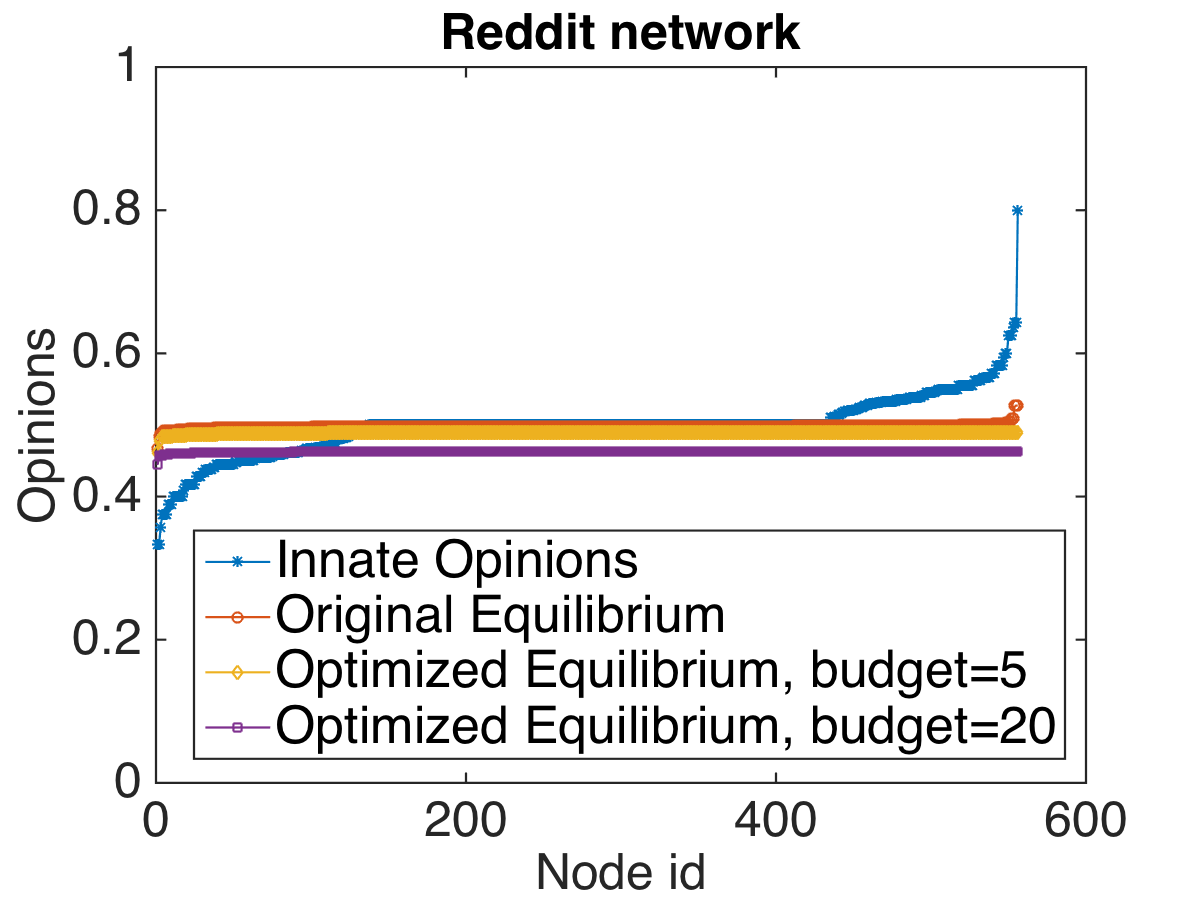} \\
(c) & (d) \\ 
\end{tabular}
\caption{\label{fig:twitterreddit} (a) Twitter network. (b) Reddit network.  (c),(d) 
Opinions per node for Twitter and Reddit respectively. Plots shows the innate opinions (*), the opinions at equilibrium with no innate  opinion intervention (o), and the opinions at equilibrium using an intervention with budget 5 (diamonds), and 20 (squares) respectively. } 
\end{figure*} 

We checked the structure of our optimal graphs, suspecting that two clusters would emerge. Specifically, we centered the innate vector $s$ around the mean, and we obtained two sets of nodes: those with innate opinions below the mean, and those with innate opinions above the mean. Intuitively, small disagreement graphs should have a strong community structure on these two sets of nodes, and small polarization graphs should have a large max cut (edges across the two groups). The optimal graphs do not exhibit such a clear two-block structure, but look closer to being random. 
Combined with our finding on random binomial graphs from the synthetic experiments, it is an interesting question whether random binomial graphs are indeed close to the optimal solution for any initial opinion vector. 

Runtime-wise, the optimization took a couple of hours for each network; observe that the total number of variables grows as $O(n^2)$. It is an open question  to come up with an efficient primal-dual procedure that takes advantage of specific features of our convex programs and Theorem~\ref{sparsifiers} to improve the overall runtime.

We also considered the problem of optimizing the vector of innate opinions using different budget values. We show how the equilibria opinions vectors look before and after intervention. Figures~\ref{fig:twitterreddit}(b)(c) show four  opinion vectors permuted  according to the sorting permutation of $s$ for Twitter and Reddit respectively: the vector of innate opinions $s$ that is sorted, the original equilibrium vector, and two optimized equilibria vectors, when the budget is 5, and 20 respectively.  We observe that nodes' opinions converge closer and closer to the average, achieving significantly lower objective values as the budget increases.

\section{Conclusion}
\label{sec:concl}

Humans tend to prefer links that minimize disagreement due to the well-known confirmation bias, among other reasons. However, this may lead to polarized communities.  In this work we introduce the notion of  the {\em polarization-disagreement index} to understand better this  trade-off between disagreement and polarization in online social networks and online social media. To the best of our knowledge, we provide the first formulation for finding an optimal topology which minimizes the sum of polarization and disagreement under a popular opinion formation model. We prove interesting facts about our objective, including the fact that there is always a graph with $O(n/\epsilon^2)$ edges that provides a $(1+\epsilon)$ approximation to the optimal objective. We also provide efficient procedures to optimize the objective. We complement our results on optimizing graph topologies, by considering a version of the problem where we optimize over innate opinions for a fixed network. Our proposed tools provide useful information about the importance of certain edges and nodes with respect to opinion formation. We provide an extensive empirical study of our results using synthetic data and two real-world datasets (Twitter, and Reddit). 

Our work opens many interesting questions. We used the popular Friedkin-Johnsen opinion formation model, which takes into account both consensus and disagreement in the underlying opinion update process. The same questions we asked here can be also asked using other opinion formation models \cite{mossel2017opinion}.  Our experimental results strongly indicate that Erd\"{o}s-R\'{e}nyi graphs are not far away from achieving the optimal polarization-disagreement index $\mathcal{I}_{G,s}$. Can we prove that Erd\"{o}s-R\'{e}nyi graphs are always a good approximation to the optimal solution? If yes, can we prove a similar result for expanders? Do links that cross communities always improve the polarization-disagreement index, or not?  Another interesting research direction is understanding how well we can approximate  the non-convex objective of Theorem~\ref{nonconvexobjective}, where a weighting is added to the polarization-disagreement index.  Furthermore, we focused on disagreement and polarization, two key phenomena in human societies. Can we develop a similar framework for optimizing other mathematical objectives for social good? 

\section*{Acknowledgements} 
Charalampos Tsourakakis would like to thank his newborn son Eftychios for the happiness he brought to his family.

\bibliographystyle{abbrv}
\bibliography{ref}

\end{document}